\documentclass[11pt]{article}

\usepackage{tikz}
\usetikzlibrary{matrix,arrows,shapes}
\usepackage{pgf}
\usepackage{url}
  \usepackage{enumerate}
  \usepackage{amsthm,amssymb,amsmath,amscd}
    \usepackage[retainorgcmds]{IEEEtrantools}
  \usepackage{eucal}
  \usepackage{mathrsfs}
  \usepackage{mathtools}
  \usepackage{upgreek}
   \usepackage{dsfont}
      \usepackage{yfonts}
  \usepackage[T1]{fontenc}
  \usepackage{bbm}
  \usepackage{geometry}
\usepackage{array}
\usepackage{booktabs}
\usepackage[retainorgcmds]{IEEEtrantools}
  \usepackage{ulem}
\definecolor{see}{RGB}{67,75,179}
\definecolor{darksee}{RGB}{42,44,148}
\definecolor{honey}{RGB}{232,180,129}
\definecolor{lighthoney}{RGB}{255,254,220}
\definecolor{citecol}{rgb}{0.5,0,0} 
\usepackage[a4paper]{hyperref}
\hypersetup
{
bookmarks=true,
bookmarksopen=true,
pdftitle="BV formalism",
pdfauthor="Katarzyna Rejzner", 
pdfsubject="BV formalism", %subject of the document
pdfmenubar=true, %menubar shown
pdfhighlight=/O, %effect of clicking on a link
colorlinks=true, %couleurs sur les liens hypertextes
pdfpagemode=None, %aucun mode de page
pdfpagelayout=SinglePage, %ouverture en simple page
pdffitwindow=true, %pages ouvertes entierement dans toute la fenetre
linkcolor=see, %couleur des liens hypertextes internes
citecolor=red!75!darkgray, %couleur des liens pour les citations
urlcolor=see %couleur des liens pour les url
}
\newcommand{\CE}{\mathfrak{CE}}  
\newcommand{\E}{\mathfrak{E}}
\newcommand{\V}{\mathfrak{V}}

\newcommand{\BV}{\mathfrak{BV}}

\newcommand{\fA}{\mathfrak{A}}

\newcommand{\D}{\mathfrak{D}}

\newcommand{\F}{\mathfrak{F}}

\newcommand{\frakg}{\mathfrak{g}}

\newcommand{\Gcal}{\mathcal{G}}  % gauge group

\newcommand{\Ccal}{\mathcal{C}}
\newcommand{\Dcal}{\mathcal{D}}
\newcommand{\Ecal}{\mathcal{E}}

\newcommand{\Ocal}{\mathcal{O}}
\newcommand{\Scal}{\mathcal{S}}

\newcommand{\Tcal}{\mathcal{T}}

\newcommand{\Ci}{\mathcal{C}^\infty} % smooth functions

\newcommand{\obj}{\mathrm{Obj}}

\newcommand{\Loc}{\mathrm{\mathbf{Loc}}}       % spacetimes
\newcommand{\Vect}{\mathrm{\mathbf{Vec}}}       % the category of locally convex topological vector
\newcommand{\WF}{\mathrm{WF}}         % wave front set
\newcommand{\id}{\mathrm{id}}               % identity
\newcommand{\supp}{\mathrm{supp}}      % support
\newcommand{\tr}{\mathrm{tr}}                 % trace
\newcommand{\loc}{\mathrm{loc}}

\newcommand{\reg}{\mathrm{reg}}

\newcommand{\pg}{\mathrm{pg}}

\newcommand{\af}{\mathrm{af}}
\newcommand{\ta}{\mathrm{ta}}
\newcommand{\gh}{\mathrm{gh}}

\newcommand{\mc}{{\mu\mathrm{c}}}
\newcommand{\ml}{\mathrm{ml}}
\newcommand{\ex}{\mathrm{ext}}

\newcommand{\NN}{\mathbb{N}}          % natural naumbers
\newcommand{\RR}{\mathbb{R}}           % real  numbers
\newcommand{\al}{\alpha}
\newcommand{\bet}{\beta}
\newcommand{\ga}{\gamma}
\newcommand{\Ga}{\Gamma}
\newcommand{\de}{\delta}
\newcommand{\De}{\Delta}

\newcommand{\la}{\lambda}
\newcommand{\La}{\Lambda}

\newcommand{\ph}{\varphi}
\newcommand{\T}{\cdot_{{}^\Tcal}}

\newcommand{\TR}{\cdot_{{}^{\TTR}}}

\newcommand{\TT}{\Tcal}
\newcommand{\TTR}{\Tcal_H}

\newcommand{\Tcirc}{{\bigcirc\!\!\!\!\!{\T}}}
\newcommand{\paqft}{{p\textsc{aqft}}}

\newcommand{\eom}{{\textsc{eom}}}
\newcommand{\qme}{{\textsc{qme}}}

\newcommand{\cme}{{\textsc{cme}}}
\newcommand{\mwi}{{\textsc{mwi}}}
\newcommand{\sst}[1]{\scriptscriptstyle{#1}}  % small font for the subscripts
\newcommand{\hinv}{*^{\!\sst{-\!1}}}             % inverse of the Hodge star
\newcommand{\minus}{\sst{-1}}   % power ^{-1}
\newcommand{\pa}{\partial}                              % partial derivative
\newcommand{\be}{\begin{equation}}
\newcommand{\ee}{\end{equation}}

\newcommand{\Lap}{\bigtriangleup}

\newcommand{\os}{\stackrel{\mathrm{o.s.}}{=}}
\newcommand{\osV}{\stackrel{\mathrm{o.s._V}}{=}}

\newcommand{\Pei}[2]{\lfloor #1, #2 \rfloor}
\newcommand{\skal}[2]{\left< #1 , #2 \right>}

\newcommand{\dgr}{{\sst\ddagger}}

\author{\null\\Katarzyna Rejzner \\
  \null\\
  \null\\
        \small{University of York}\\
    \small{York YO10 5DD}\\ 
\small{\texttt{kasia.rejzner@york.ac.uk}}}

\title{Remarks on local symmetry invariance in\\ perturbative algebraic quantum field theory}
\begin{document}
\date{}
 \maketitle

  \theoremstyle{plain}
  \newtheorem{df}{Definition}[section]
  \newtheorem{thm}[df]{Theorem}
  \newtheorem{prop}[df]{Proposition}
  \newtheorem{cor}[df]{Corollary}
  \newtheorem{lemma}[df]{Lemma}
  \newtheorem{exa}[df]{Example}
  
  \theoremstyle{plain}
  \newtheorem*{Main}{Main Theorem}
  \newtheorem*{MainT}{Main Technical Theorem}

  \theoremstyle{definition}
  \newtheorem{rem}[df]{Remark}
 \theoremstyle{definition}
  \newtheorem{ass}{\underline{\textit{Assumption}}}[section]

%\newpage
%*************************************************************
\abstract{We investigate various aspects of invariance under local symmetries in the framework of perturbative algebraic quantum field theory (pAQFT). Our main result is the proof that the quantum Batalin-Vilkovisky  (BV) operator, on-shell, can be written as the commutator with the interacting BRST charge. Up to now, this was proven only for a certain class of fields in quantum electrodynamics and in Yang-Mills theory. Our result is more general and it holds in a wide class of theories with local symmetries, including general relativity and the bosonic string. We also comment on other issues related to local gauge invariance and, using the language of homological algebra, we compare different approaches to quantization of gauge theories in the pAQFT framework.}
\tableofcontents
\markboth{Contents}{Contents}
\section{Introduction}
Perturbative algebraic quantum field theory ({\paqft}) is a mathematical framework developed during the last two decades to study problems in perturbative renormalization. It proved to be very useful in constructing models in quantum field theory on curved spacetimes, because the operator algebraic approach allows one to separate the construction of the algebra of observables from the construction of a state. Research in {\paqft} is focused on two main problems: developing methods for renormalization on general globally hyperbolic backgrounds \cite{BF0,BFV,FR,FR3,H,HW,HW2,HW3,HW5}, and identifying algebraic structures appearing in perturbative renormalization on Minkowski spacetime \cite{BDF,DF,Duetsch:2000nh,DF02,DF04,Kai,FR,FR3,Rej}. An important step towards the consistent axiomatic framework for QFT on curved spacetimes was introducing the notion of \textit{general local covariance} \cite{BFV,HW}. In \cite{BFV} this notion is formulated in the language of category theory. The axiomatic framework proposed in \cite{BFV} is a generalization of the Haag-Kastler axioms \cite{HK} of local quantum field theory. In the {\paqft} approach, the main tool  we use to investigate the algebraic structures appearing in renormalization is the Epstein-Glaser \cite{EG} method. It allows us to prove the existence of renormalized quantities (time-ordered products) without having to manipulate ill defined objects in the intermediate steps of quantization.

Initially the {\paqft} framework was developed for scalar fields, but recently there has been a lot of progress in constructing more complicated models. In particular, quantum electrodynamics (QED) was discussed in \cite{DF99} and Yang-Mills theory was discussed in \cite{Boas,H}. A general framework which deals with arbitrary theories with local symmetries was subsequently proposed in \cite{FR,FR3,Rej11b}. This setting makes use of the Batalin-Vilkovisky formalism, which relies on homological algebra methods. In 
\cite{FR,FR3} these algebraic tools are refined by introducing functional-analytic aspects and generalizing the BV formalism to infinite dimensional spaces. In \cite{FR3} a general quantization scheme for gauge theories is proposed and some comparison with the approach of \cite{H} is made. In the present work we want to continue this line. We discuss various aspects of local gauge invariance in {\paqft}, pointing out differences and common features of existing approaches. The framework proposed in \cite{FR,FR3} is a very convenient tool for such analysis, since it is general and flexible enough. We focus our discussion on two problems: 
 general formulation of consistency conditions that have to be satisfied by the deformed $\star$-product in order to be compatible with structures appearing in the BV formalism, and the definition and intrinsic meaning of the free and the interacting BRST charge.
Our main result is the proof that the interacting BRST charge $R_V(Q)$ ($R_V$ denotes the derivative of the relative S-matrix) generates on-shell the quantum BV operator $\hat{s}=s-i\hbar\Lap_V$, defined in \cite{FR3}, i.e.
\be\label{main:1}
\frac{i}{\hbar}[R_V(F),R_V(Q)]_\star=R_V(\hat{s} F)
\ee
holds for local $F$, modulo the free equations of motion. Our proof generalizes results obtained in \cite{H}, since we do not restrict to $F$'s for which the renormalized BV Laplacian $\Lap_V(F)$ vanishes. Moreover, our result can be applied to a larger class of theories, including gravity \cite{BFR} and the bosonic string \cite{BRZ}.
 Using the interacting star product one can formulate this result also as
\be\label{main:2}
\frac{i}{\hbar}[V,Q]_{\star_V}=\hat{s}F\,,
\ee
modulo interacting equations of motion.

The paper is organized as follows: in the first section we construct the classical theory in the framework of {\paqft}, in the second section we perform the quantization and in the last section we discuss the BRST charge. The content of the first section is essentially a brief summary of the formalism introduced in \cite{FR}, but we formulate it here in the language of graded differential geometry. We also prove some properties of the renormalized BV Laplacian $\Lap_V$. The second section starts with a detailed discussion of consistency conditions that the star product has to fulfill in order to be compatible with the free BRST operator. We show how this relates to the free quantum master equation {\qme}, and we identify intrinsic reasons for the conditions to arise. Such consistency conditions are necessary, if one works with the linearized BV operator. This is an argument in favor of the approach proposed by K.~Fredenhagen and myself in \cite{FR3}, where we work with the full BV operator, instead. In  \cite{FR3}, we construct interacting fields from free ones by means of the intertwining map $R_V$; we include $\theta_0$ (the part of the Lagrangian which generates the free BRST operator) in the interaction term $V$, and the free action $S_0$ doesn't contain antifields. In subsection \ref{changing} we show how the theory whose starting point action is $S_0$ relates to the theory whose starting point action is $S_0+\theta_0$. This concludes the second section of the paper. In the third section we discuss the BRST charge. We prove relations \eqref{main:1} and \eqref{main:2}, and discuss in detail differences between approaches to gauge theory quantization taken in \cite{DF99} and \cite{H}.
\section{Classical field theory}
\subsection{Kinematical structure}
We start with the kinematical structure. Let  $M$  be an oriented, time-oriented globally hyperbolic spacetime. We associate to $M$ the space  $\E(M)$, of field configurations of the theory. $\E(M)$ describes the physical content of the theory, i.e. specifies what kind of objects the theory is describing. The results of the present work can be applied to a very general class of theories, including Yang-Mills theory and gravity.  We only assume that the configuration space is a space of smooth sections of some natural vector bundle $E\xrightarrow{\pi} M$ with fiber $V$ over $M$, i.e. $\E(M)=\Gamma(E)\equiv\Gamma(M,V)$. Let $\E_c(M)$ denote the space of compactly supported configurations and $\D(M)\doteq\Ci_0(M,\RR)$. A classical measurement associates, to a configuration in $\E(M)$, a real number. Therefore it is natural to identify classical observables with functionals $F:\E(M)\to\RR$. 
We require these functionals to be smooth in the sense of calculus on locally convex topological vector spaces. Let us briefly recall the relevant definitions. The derivative of $F$ at $\ph\in\E(M)$ in the direction of $\psi\in\E(M)$ is defined by
\be\label{de}
\left<F^{(1)}(\ph),\psi\right> \doteq \lim_{t\rightarrow 0}\frac{1}{t}\left(F(\ph + t\psi) - F(\ph)\right)\,,
\ee
whenever the limit exists. The function $F$ is called differentiable at $\ph\in\E(M)$ if $\left<F^{(1)}(\ph),\psi\right>$ exists for all $\psi \in \E(M)$. It is called continuously differentiable if it is differentiable at all points of $\E(M)$ and
$F^{(1)}(.):\E(M)\times \E(M)\rightarrow \RR, (\ph,\psi)\mapsto \left<F^{(1)}(\ph),\psi\right>$
is a continuous map. It is called a $\Ccal^1$-map if it is continuous and continuously differentiable. Higher derivatives are defined in a similar way.  The continuity condition for derivatives implies that $F^{(n)}(\ph)\in \Gamma'(M^n,V^{\otimes n})$ holds  for all $\ph\in\E(M)$, $n\in\NN$, so  $F^{(n)}(\ph)$ is a distributional section with compact support.

An important property of a functional $F$ is its spacetime support. It is defined by
\begin{align}\label{support}
\supp\, F=\{ & x\in M|\forall \text{ neighbourhoods }U\text{ of }x\ \exists \ph_1,\ph_2\in\E(M), \supp\, \ph_2\subset U 
\\ & \text{ such that }F(\ph_1+\ph_2)\not= F(\ph_1)\}\ .\nonumber
\end{align}
Another crucial property of a functional is \textit{the locality}. According to the standard definition it means that the functional $F$ is of the form:
\[
F(\ph)=\int\limits_M f(j_x(\ph))\,d\mu(x)\,,
\]
where $f$ is a function on the jet space over $M$, $j_x(h)=(x,\ph(x),\pa \ph(x),\dots)$ is the jet of $\ph$ at the point $x$ and $d\mu(x)$ denotes the invariant measure on $M$ induced by the metric. The space of compactly supported smooth local functions $F:\E(M)\to\RR$ is denoted by $\F_\loc(M)$. The algebraic completion of $\F_\loc(M)$ with respect to the pointwise product
\be\label{prod}
F\cdot G(h)=F(h)G(h)
\ee
is a commutative algebra $\F(M)$, consisting of finite sums of finite products of local functionals. We call this space \textit{the algebra of multilocal functionals}. Both $\F_\loc$ and $\F$ are covariant functors from $\Loc$ (the category of globally hyperbolic oriented and time-oriented spacetimes with causal isometric, (time)-orientation preserving embeddings as morphisms) to the category $\Vect$ of locally convex vector spaces.
\subsection{Dynamics and symmetries}
Following \cite{BDF}  we introduce the dynamical principle by means of a generalized Lagrangian. Let  $L$ be a natural transformation between the functor of test function spaces $\D$, and the functor $\F_\loc$. For each $M\in\obj(\Loc)$ we have a morphism $L_M:\D(M)\rightarrow \F_\loc(M)$ in $\Vect$. $L$ is a generalized Lagrangian if all these morphisms, numbered by objects of $\Loc$, satisfy
\be\label{L:supp}
\supp(L_M(f))\subseteq \supp(f)
\ee
and the additivity rule 
\be\label{L:add}
L_M(f+g+h)=L_M(f+g)-L_M(g)+L_M(g+h)\,,
\ee
where $f,g,h\in\D(M)$ and $\supp\,f\cap\supp\,h=\emptyset$.  
The action $S(L)$ is defined as an equivalence class of Lagrangians  \cite{BDF}, where two Lagrangians $L_1,L_2$ are called equivalent $L_1\sim L_2$  if
\be\label{equ}
\supp ({L_{1}}_M-{L_{2}}_M)(f)\subset\supp\, df\,, 
\ee
for all $f\in\D(M)$. Let us fix a Lagrangian $L_{\textrm{ph}}$, defining our physical theory. 
To derive the equations of motion, we follow the approach of \cite{BDF} and define the Euler-Lagrange derivative of $S_{\textrm{ph}}$ as a natural transformation ${S'_{\textrm{ph}}}:\E\to\D'$ given by
\be\label{ELd}
\left<({S_{\textrm{ph}}}')_M(\ph),\psi\right>=\left<(L_{\textrm{ph}})_M(f)^{(1)}(\ph),\psi\right>\,,
 \ee
 where $f\equiv 1$ on $\supp \psi$. The equation
\be
 S_{\textrm{ph}}'(\ph)\equiv0\,.\label{eom}
\ee
is called \textit{the equation of motion} (\eom). The space of solutions of \eqref{eom} is a subspace of $\E(M)$ denoted by $\E_S(M)$. In the physics literature one calls functionals on $\E_S(M)$ \textit{on-shell functionals}. Analogously, equalities that hold for functions restricted to $\E_S(M)$ are called \textit{on-shell} equalities. 

Local symmetries of the action $S_{\textrm{ph}}$ are described by certain vector fields on $\E(M)$. We want to consider only variations in the directions of compactly supported configurations, so the corresponding space of vector fields can be identified with
 \[
 \V(M)=\{X:\mathfrak{E}(M)\to\E_c(M)| X\text{ smooth with compact support} \}\,.
 \]
% In more precise terms this is the space of vector fields on $\mathfrak{E}(M)$ considered as a manifold modeled over $\mathfrak{E}_c(M)$\footnote{For more details about infinite dimensional manifolds, see \cite{Neeb,Michor}.}. 
 $X\in\V(M)$ is a symmetry if
\[
(\partial_X S_{\textrm{ph}})(\ph)\equiv 0,\ \forall\ph\in\E(M)\,,
\]
where
\[
(\partial_X S_{\textrm{ph}})(\ph)\doteq \left<(L_{\textrm{ph}}(f))^{(1)}(\ph),X(\ph)\right>\,,\quad f\equiv 1\ \textrm{on}\ \supp X\,.
\]
The space of all symmetries of the given action has a structure of an infinite dimensional Lie algebra. In case of gauge theories and gravity this space can be characterized in a simple way. There exists a space of smooth sections of some vector bundle $\frakg_c(M)=\Gamma(M,g)$ which carries a structure of a Lie algebra and there is a Lie algebra morphism $\rho:\frakg_c(M)\rightarrow \Gamma_c (T\E(M))$ such that
every symmetry $X$ can be expressed as $X(\ph)=\rho(\xi_{\ph})(\ph)+I$, where $\xi_{\ph}\in \frakg_c(M)$ and $I$ is a trivial symmetry  (a vector field that vanishes identically on $\E_S(M)$). In cases which we consider, $\frakg_c(M)$ arises as a Lie algebra of an infinite dimensional Lie group $\Gcal(M)$, called the gauge group. Since we work on a fixed background, from now on we will keep the argument ``$(M)$'' implicit, whenever this doesn't create confusion. 

In order to quantize the theory we need a way to characterize the space of functionals invariant under the symmetries of $S_{\mathrm{ph}}$. In \cite{FR} it was shown how to achieve this using an appropriate extension of the Batalin-Vilkovisky (BV) formalism. In the first step one constructs the space of alternating multilinear forms (the so-called ghosts) on  $\frakg_c$ with values in $\F$. In addition, we require multilocality and compact support, so we consider the space $\CE\doteq\Ci_\ml(\E,\Lambda\frakg')$. It is a graded algebra and the corresponding grading is called the pure ghost number $\#\pg$. In the topology described in \cite{FR}, $\CE$ is the completion of $\F\otimes\Lambda\frakg'$ and therefore we interpret it as the space of functions on the infinite dimensional graded manifold $\overline{\E}\doteq\E\oplus\frakg[1]$ (the number in square bracket denotes the shift in degree).

The Chevalley-Eilenberg  differential $\gamma_{\mathrm{ce}}$ is defined in the standard way \cite{ChE} and can be identified with the exterior derivative on the space of gauge equivariant forms on the gauge group $\Gcal$. The 0-th cohomology of $\gamma_{\mathrm{ce}}$ is the space of gauge invariant functionals\footnote{In the case of gravity one needs to define this cohomology not on the level of functionals, but fields, i.e. natural transformations from $\D$ to $\F$. This is discussed in details in \cite{FR}.}. The Batalin-Vilkovisky algebra $\BV$ is the graded symmetric tensor algebra of graded derivations of $\CE$. Again we require that elements of $\BV$ considered as smooth maps on $\E$ with values in a certain graded algebra are multilocal and compactly spacetime supported. The resulting space is
\[
\BV\doteq\Ci_\ml\big(\E,\La\E_c\widehat{\otimes}\La{\frakg}'\widehat{\otimes}S^\bullet \frakg_c\big)\,.\label{BVfix}
\]
This is again a completion of $\F\otimes\La\E_c\otimes\La{\frakg}'\otimes S^\bullet \frakg_c$, so
we can interpret the elements of the above space as functionals on 
\[
\E[0]\oplus\frakg[1]\oplus \E_c'[-1]\oplus\frakg'_c[-2]\,,
\]
which is the odd cotangent bundle $\Pi T^*(\overline{\E})$ of the extended configuration space $\overline{\E}=\E[0]\oplus\frakg[1]$,
where the manifold structure on $\E\oplus\frakg[1]$ is defined by the basis of neighborhoods with the topology of $\E_c\oplus\frakg_c$. For simplicity we denote by $\ph^\al$ an element of $\overline{\E}$ and the index $\al$ runs through all the physical and ghost indices. The full field multiplet will be denoted by $\ph$ and an evaluation functional on $\overline{\E}$ will be written as $\Phi_x^\al$. Functions on the graded vector space $\Pi T^*(\overline{\E})$ are the graded multivector fields on $\overline{\E}$ and we can write  example elements of $\BV$ in the form
\be
\label{Polynom}
F=\int d\mu(x_1,\dots,x_{m}) f_F(x_1, \dots ,x_{m})\Phi_{x_1}\!\dots\Phi_{x_k}  \tfrac{\delta}{\delta \ph(x_{k+1})} \dots   \tfrac{\delta}{\delta \ph(x_{m})}\,,
\ee
where $d\mu(x_1,\dots,x_n)$ denotes the measure $d\mu(x_1)\dots d\mu(x_n)$, we keep the summation over  the indices $\alpha$ implicit and the product denoted by juxtaposition is the graded associative product of $\BV$. We can treat the functional derivatives $\tfrac{\delta}{\delta \ph^\al(x)}$ as ``basis'' on the fiber $\E_c'[-1]\oplus\frakg'_c[-2]$ and we denote them by $\ph_\al^\dgr$.
In the physics literature they are called \textit{antifields}.
%where $\mathbf{h}_{x_i}$, $C_{x_j}$ are the evaluation functionals on $\E(M)$ and $\frakg(M)$ respectively. In this section we use consequently the notation in which $\mathbf{h}_{x_i}$ and the uppercase letters $C_{x_j}$ are the evaluiation functionals and lower case letters $h$, $c$ are the variables, i.e. elements of $\E(M)$, $\frakg(M)$.
%Similarly  $\mathbf{h}^\dgr_{y_i}$, $C^\dgr_{y_j}$ are the antifields, identified in our formalism with functional derivatives: $\mathbf{h}^\dgr_{y_i}\equiv\frac{\delta}{\delta h(y_i)}$, $C^\dgr_{y_j}\equiv\frac{\delta}{\delta c(y_j)}$. 
In the above formula $f_F$ is a distribution with the wavefront set orthogonal to the total diagonal, and with the support that is compact and is contained in the product of partial diagonals. Later we will extend our discussion to more singular objects. In the appropriate topology (more details in \cite{FR}) elements of the form \eqref{Polynom} are dense in  $\BV$, so we can often restrict our discussion to such elements, without the loss of generality. Functional derivative with respect to an odd variable or an antifield can be defined on elements \eqref{Polynom} as the left derivative and extended to $\BV$ by continuity. We will always assume that $\tfrac{\delta}{\delta \ph^\al(x)}$, $\tfrac{\delta}{\delta \ph_\al^\dgr(x)}$ are the left derivatives, unless stated otherwise.

The algebra $\BV$ has two gradings: ghost number $\#\gh$ and antifield number $\#\af$. Functionals of physical fields have both numbers equal to 0. Functionals of ghosts have a $\#\gh=\#\pg$ and  $\#\af=0$. All the vector fields have a non-zero antifield number and $\#\gh=-\#\af$.  The space $\BV$ seen as the space of graded multivector fields is equipped with a graded generalization of the Schouten bracket $\{.,.\}$, called in this context \textit{the antibracket}. The space of on-shell functionals is characterized by means of the Koszul operator. It can be written as the antibracket with the physical action $S_{\mathrm{ph}}$,
\be\label{Koszul}
\delta_{\mathrm{ph}} F=\{F,L_{\mathrm{ph}}(f)\},\ F\in \BV,\,f\equiv 1\ \textrm{on }\supp\, F\,.
\ee
To simplify the notation we often write  $\delta_{\mathrm{ph}} F=\{F,S_{\mathrm{ph}}\}$, instead of \eqref{Koszul}. In analogy to \eqref{Koszul} one finds a generalized Lagrangian $\theta_{\mathrm{ce}}$, which implements the Chevalley-Eilenberg differential,  $\gamma_{\mathrm{ce}} F=\{.,\theta_{\mathrm{ce}}\}$. The total BV differential is the sum of the Koszul-Tate differential and the Chevalley-Eilenberg differential:
 \[
 s_{BV}F\doteq\{F,S+\theta_{\mathrm{ce}}\}\,,
 \]
which satisfies $s_{BV}^2=0$  and the 0-th cohomology of $(\BV,s_{BV})$ is the space of gauge invariant on-shell multilocal functionals: $\BV^{\,ph}=H^0(\BV, s_{BV})$. 

In the next step we introduce the gauge fixing. Often one needs to extend the BV complex by adding auxiliary fields, for example antighosts $\bar{C}$ and Nakanishi-Lautrup fields $B$. One obtains a new extended configuration space and the corresponding extended space of multilocal functionals, denoted by $\BV$. The full multiplet is denoted by  $\ph\in\overline{\E}$ and it is a section of some graded vector bundle $\overline{E}$ over $M$. We assume that on $\overline{\E}$ there exists a duality (in the sense of vector spaces, not graded vector spaces) $\left<.,.\right>_{\overline{\E}}$, which allows to embed $\overline{\E}$ in $\overline{\E}_c'$ (in the example of the electromagnetic field this is just the Hodge duality).

To fix the gauge, we perform first an automorphism $\alpha_\Psi$ of $\BV$ which leaves the antibracket invariant (see \cite{FR} for details). Performing this automorphism formally means replacing antifields $\ph_\al^\ddagger(x)$ by  $\ph_\al^\ddagger(x)+\tfrac{\delta \Psi}{\delta\ph^\al(x)}$, where $\Psi\in\BV(M)$ is a functional with $\#\gh=-1$, called the gauge fixing fermion. The action $S+\theta_{\mathrm{ce}}$ is transformed into a new action $S_\ex=\al_\Psi(S+\theta_{\mathrm{ce}})$. One introduces also a new grading, which is sometimes called the \textit{total antifield number} $\#\ta$. It is equal to 1 for vector fields on $\overline{\E}$, irrespective of their antifield number and is equal to 0 for functions on $\overline{\E}$. Note that a functional can have a non-zero $\#\ta$, but have $\#\af=0$. This is the case in the non-minimal sector in QED, Yang-Mills or general relativity; the antifield $\bar{C}^\ddagger$ of the antighost $\bar{C}$ has $\#\af=0$, but, under the identification $\bar{C}^\ddagger\equiv\frac{\delta}{\delta \bar{C}(x)}$, $\bar{C}^\ddagger$ is a derivation, so  $\#\ta(\bar{C}^\ddagger)=1$. This subtlety plays a role in the discussion of the BRST charge presented in section \ref{free:charge}. The gauge-fixing fermion has to be chosen in such a way that \textit{the gauge fixed action} $S$ (the $\#\ta=0$ part of $S_\ex$) has a well posed Cauchy problem (see \cite{FR} for details). %\marginpar{interpret this in the language of homotopy: Frederic}.
The transformed BV differential is given by 
\[
s=\al_\Psi\circ s_{BV}\circ\al_\Psi^{-1}=\{.,S_\ex\}
\]
and we can expand it with respect to the total antifield number $\#\ta$,
\be\label{ta:expansion}
s=\gamma+\delta\,,
\ee
where the differential $\delta$ is the Koszul operator for the field equations derived from $S$ and
 $\gamma$ is generated by $\theta=S_\ex-S$. In this context $\gamma$ is usually called \textit{the gauge-fixed BRST operator}. The uniqueness of the Cauchy problem solution for the \eom's derived from $S$ implies that $(\BV,\delta)$ is a resolution.
 \section{Quantization}
\subsection{Free theory}\label{free:theory}
From the point of view of quantization it is convenient to split the gauge fixed action $S$ into a quadratic part and the remainder, called \textit{the interaction term}. We perform the Taylor expansion
\be\label{Taylor1}
L(f)(\ph_0+\ph)=L(f)(\ph_0)+\left<L(f)^{(1)}(\ph_0),\ph\right>+\frac{1}{2}
\left<L(f)^{(2)}(\ph_0);\ph,\ph\right>+\dots\,,
\ee
where $\left<L(f)^{(1)}(\ph_0),\ph\right>\doteq\sum_\al\left<\frac{\delta}{\delta \ph^\al}(L(f))(\ph_0),\ph^\al\right>$ and
$\al$ runs through all the indices of the field configuration multiplet $\ph$. The first term of the above expansion is just a constant, the second one vanishes if we choose the background configuration $\ph_0$ to be a solution of \eom's. We denote this term by $L_{\textrm{lin}}$. The third term, denoted by $L_0$, is the quadratic part  of the gauge fixed Lagrangian. Let $S_0$ denote the action corresponding to $L_0$. 
%From symmetry reasons it is convenient to use another representative of the equivalence class of $S_0$, namely:
%\[
%L_0(f)=\int \!f (x)\Phi^\al_x P_{\al\beta}f(x)\Phi^\beta_x\, d\mu(x)\,.
%\]
%or in a more compact notation:
%\[
%L_0(f)=\left<f\Phi,P(f\Phi)\right>
%\]
We obtain an expansion:
\[
S=S(\ph_0)+S_{\textrm{lin}}+S_0+\dots\,.
\]
The Euler-Lagrange derivative of $S_0$ induces the Euler Lagrange operator operator $P:\overline{\E}\rightarrow \overline{\E}_c'$. Moreover, we assume that the image of $P$ is contained in $\overline{\E}$ (elements of $\overline{\E}$ are identified with distributions by means of $\left<.,.\right>_{\overline{\E}}$) and that
%whose action on configurations $\ph\in\overline{\E}$ can be written as
%\[
%(P\ph)_\al(x)=\sum_{\beta}P_{\al\beta}(x)\ph^\beta(x)\,,
%\]
%where each component $P_{\al\beta}$ is a differential operator on $\Ci(M,\RR)$. 
%Assume that 
the gauge fixing is done in such a way that $P$ is normally hyperbolic (for gauge theories and gravity this was shown in \cite{FR}, the bosonic string was studied in \cite{BRZ}). This implies that $P$ has unique retarded and advanced propagators $\Delta^{A/R}$, i.e. the relations
\begin{align*}
P\circ\Delta^{A/R}&=\id_{\overline{\E}_c}\,,\\
\Delta^{A/R}\circ P\big|_{\overline{\E}_c}&=\id_{\overline{\E}_c}
\end{align*}
hold  and $\Delta^{A/R}$ fulfill the support properties
\begin{align*}
\supp(\Delta^R)&\subset\{(x,y)\in M^2| y\in J^-(x)\}\,,\\
\supp(\Delta^A)&\subset\{(x,y)\in M^2| y\in J^+(x)\}\,.
\end{align*}
The causal propagator is defined as $\Delta=\Delta^R-\Delta^A$. Let us denote $S_1=S^\ex-S_0$. In the first step we will quantize the free theory, i.e. the one defined by the free action $S_0$. $S_1$ is the full interaction term, with antifields included. We expand now the gauge-fixed BRST differential $\gamma$. Since $\theta(f)$ depends also on antifields, we have to take them into account as well. Actions we consider are polynomial in antifields, so the left derivative with respect to $\ph^\ddagger$ makes sense and we can expand
%Let us first write
%\[
%\theta(f)(\ph,\ph^\dgr)=\int\!\! f(x)\,{\theta}_x(\ph)\,\ph^\dgr(x) d\mu_g(x)\,.
%\]
%Formally we can interpret ${\theta}_x(\ph)$ as the functional derivative with respect to antifields:
%\[
%f(x){\theta}_x(\ph)=\frac{\delta \theta_{M}(f)}{\delta \ph^\dgr(x)}(\ph,0)\,.
%\]
%With this notation we can now expand $\theta$ with respect to fields $\ph^\al$ and antifields:
\begin{multline*}
\theta(f)(\ph_0+\ph,\ph^\dgr)=\left<\frac{\delta \theta(f)}{\delta \ph^\dgr}(\ph_0,0),\ph^\dgr\right>+\frac{1}{2}\left<\frac{\delta^2 \theta(f)}{\delta \ph\delta\ph^\dgr}(\ph_0,0);\ph,\ph^\dgr\right>+\\
+\frac{1}{6}\left<\frac{\delta^3 \theta(f)}{\delta \ph^2\delta\ph^\dgr}(\ph_0,0);\ph,\ph,\ph^\dgr\right>+...\,.
\end{multline*}
Without any restrictions on the physical component of the background configuration $\ph_0$, we can choose $\ph_0$ in such a way that the first term in the above expansion vanishes and we obtain
\[
\theta=\theta_0+\theta_1+\dots\,.
\]
The first nontrivial term, denoted by  $\theta_0$, generates the free BRST differential.
Its derivative $\frac{{\delta^l}^2 \theta_0(f)}{\delta\ph^\sigma\delta\ph_\al^\dgr}$ is an element of $\overline{\E}\otimes \overline{\E}'$
and it induces a differential operator $K:\overline{\E}\rightarrow \overline{\E}$. In local coordinates we can write the second derivative of $\theta_0(f)$ as $f(y) K^{\al}_{\phantom{\al}\sigma}(x')\delta(y-x')$ and
\be\label{gamma0}
\gamma_0=\sum_{\sigma,\al}{K}^{\al}_{\ \sigma}\Phi_x^\sigma\frac{\delta}{\delta\ph^\al(x)}\,.
\ee
\begin{exa}[Free electromagnetic field]\label{ex1}
{\small As an example consider a free electromagnetic field described by the Lagrangian $L_{EM}(f)(A)=-\frac{1}{2}\int_M (F\wedge *F)f$ where $F=dA$ and $A\in\Gamma(T^*M)\equiv \Omega^1(M)$. The extended configuration space is $\overline{\E}(M)=\Omega^1(M)\otimes \Ci(M)\otimes \Ci(M)[1]\otimes \Ci(M)[-1]$ and an element of this space can be written as a quadruple $\ph=(A,B,C,\bar{C})$. Let us denote by $\left<.,.\right>$ the duality between $\overline{\E}(M)$ and $\overline{\E}'(M)$ and by $\left<.,.\right>_g$ the duality on the space of $p$-forms induced by the metric, i.e. $\left<u,v\right>_g\doteq\int_M u\wedge *v$, $u,v\in \Omega^p(M)$.
The extended Lagrangian takes the form
\begin{multline*}
L_M(f)(\ph)=-\tfrac{1}{2}\left<fF,F\right>_g-\left<fdC,\tfrac{\delta}{\delta A}\right>+i\left<fB,\tfrac{\delta}{\delta \bar{C}}\right>+\\-i\left<fd\bar{C},d C\right>_g-\left<fB,\hinv d*\!A-\tfrac{1}{2}B\right>_g\,.
\end{multline*}}
Operators $P$ and $K$, written in the basis $(A,B,C,\bar{C})$, take the form
\[
P=\left(\begin{array}{cccc}
\delta d&d&0&0\\
\delta&-1&0&0\\
0&0&0&i\delta d\\
0&0&-i\delta d
\end{array}\right)\,,\quad
K=\left(\begin{array}{cccc}
0&0&d&0\\
0&0&0&0\\
0&0&0&0\\
0&i&0&0
\end{array}\right)\,.
\]
\end{exa}
The classical master equation (\cme) yields $2\{\theta,S\}+\{\theta,\theta\}\sim0$, where $\sim$ is defined by \eqref{equ}. Expanding this relation in powers of field configurations, we obtain in particular (in the  second order in $\ph$ and in the 0th order in $\ph^\ddagger$)
\be\label{fullCME}
2\{\theta_0,S_0\}+\{\theta_1,S_{\textrm{lin}}\}\sim0\,.
\ee
The first two terms of this identity correspond to classical master equation for the free action $S_0+\theta_0$. Since $\theta_0$ has to be even with respect to the ghost grading and $\#\gh(\ph^\ddagger_\al)=-\#\gh(\ph^\al)-1$, we obtain $\tr K=0$ and $\{\theta_0,\theta_0\}=0$. Therefore, the classical master equation of the free theory can be simply expressed as:
\be\label{freeCME}
\{\theta_0,L_0\}\sim 0\,.
\ee
and it implies that the free action $S_0$ is invariant under the free BRST operator $\gamma_0$. Relation \eqref{freeCME} is compatible with \eqref{fullCME} only if $\{\theta_1,L_{\textrm{lin}}\}\sim0$, which is true if $\ph_0$ solves the equations of motion (i.e. $\ph_0$ is on-shell). 

We will show in subsection \ref{gauge:inv} that, in order to construct the interacting theory starting from the action $S_0+\theta_0$, one has to impose the free {\cme} and this implies that $\ph_0$ has to be a solution of the equations of motion. However, if we start from the theory with the free action $S_0$, there are no \textit{a priori} reasons to choose $\ph_0$ to be on-shell. It could, nevertheless, happen that, for construction of states, one needs to impose restrictions on $\ph_0$.

The general construction of the interacting theory, starting from $S_0$ as a free action, was performed in  \cite{FR3}. Let us recall briefly the main ideas.
The classical linearized theory is constructed by introducing the Peierls bracket given by (to simplify the sign convention we use both the right and the left derivative):
\begin{equation*}
\Pei{F}{G} = \sum_{\al,\beta} \skal{\frac{\delta^r F}{\delta\ph^\al}}{{\De}^{\al\beta}\frac{\delta^l G}{\delta\ph^\beta}},
\end{equation*}
where $F, G \in\overline{\F}_\mc(M)$ are microcausal elements of $\BV$, i.e. they are smooth, compactly supported and their derivatives (with respect to both $\ph$ and $\ph^\dgr$) satisfy the WF set condition:
 \be\label{mlsc}
\WF(F^{(n)}(\ph,\ph^\dgr))\subset \Xi_n,\quad\forall n\in\NN,\ \forall\ph\in\overline{\E}(M)\,,
\ee
where $\Xi_n$ is an open cone defined as 
\be\label{cone}
\Xi_n\doteq T^*M^n\setminus\{(x_1,\dots,x_n;k_1,\dots,k_n)| (k_1,\dots,k_n)\in (\overline{V}_+^n \cup \overline{V}_-^n)_{(x_1,\dots,x_n)}\}\,,
\ee
The space of compactly supported vector-valued distributions on $M^n$ with the WF set contained in $\Xi_n$ will be denoted by $\Ecal'_{\Xi_n}(M^n,V)$, where $V$ is some finite dimensional vector space. The space of microcausal elements of the BV complex $\BV_{\mc}$ is equipped with a topology $\tau_\Xi$ induced by the H\"ormander topology, as defined in \cite{FR}: 
The quantized algebra of the free fields is constructed by means of the deformation quantization of the classical algebra $(\BV_{\mc},\Pei{.}{.})$. To this end, we equip the space of formal power series $\BV_{\mc}[[\hbar]]$ with a noncommutative star product which corresponds to the operator product of quantum observables. For this construction one needs Hadamard parametrices. A Hadamard parametrix $\omega$ is here understood as a matrix with rows and columns numbered by the indices $\al$ of the field configuration multiplet $\ph$ and with entries in $\Dcal'(M^2)$ which fulfill
\begin{IEEEeqnarray}{rCl}\label{parametrix}
  \omega^{\alpha \beta}(x,y) - (-1)^{|\ph^\alpha| |\ph^\beta|} \omega^{\beta \alpha}(y,x)& =& i \Pei{\ph^\alpha(x)}{\ph^\beta(y)},\IEEEyessubnumber\label{classical:limit}\\
 \sum_\bet P_{\al\beta}(x) \omega^{\beta \gamma}(x,y) & =& 0\ \textrm{mod }\Ci\textrm{ function},\IEEEyessubnumber\label{field:eq}\\
 \WF(\omega^{\alpha \beta}) & \subset &C_+,\IEEEyessubnumber\label{WFset}\\
 \overline{\omega^{\alpha \beta}(x,y)} & =& \omega^{\beta \alpha}(y,x).\IEEEyessubnumber\label{hermitian}
\end{IEEEeqnarray}
%Here $\phi^\alpha \in \{ \varphi, b, c, \bar c \}$, $\betrag{\alpha}$ is the grade of $\phi^\alpha$,
By $ C_+$ we denoted the following subset of the cotangent bundle $ T^*M^2$:
\[
 C_+ = \{ (x_1, x_2; k_1, - k_2) \in T^*M^2 \setminus \{ 0 \} | (x_1; k_1) \sim (x_2; k_2), k_1 \in \bar V^+_{x_1} \},
\]
where $(x_1; k_1) \sim (x_2; k_2)$ if there is a lightlike geodesic from $x_1$ to $x_2$ to which $k_1$ and $k_2$ are coparallel.
If we replace the condition (\ref{field:eq}) by a stronger one 
 \begin{equation}\label{field:eq:s}
\sum_\bet P_{\al\beta}(x) \omega^{\beta \gamma}(x,y) =0\,,
 \end{equation}
 then the Hadamard parametrix becomes a Hadamard 2-point function. Assume that on the background manifold $M$, there exists a quasifree Hadamard state and write the corresponding 2-point function in the form $\omega=\frac{i}{2}\Delta+H$. The Feynman-like propagator is defined as $H_F=i\Delta_D+H$, where $\Delta_D\doteq \tfrac{1}{2}(\Delta_R+\Delta_A)$ is the Dirac propagator.  Let $\al_H\doteq e^{\frac{\hbar}{2}\Gamma_H}$ be a map defined on regular functionals $\BV_\reg(M)$ (i.e. functionals satisfying  the wavefront set condition $\WF(F^{(n)}(\ph,\ph^\dgr))=\varnothing$ for all $\ph$), where 
\[
\Ga_{H} \doteq \sum_{\al,\beta}\left<{H}^{\al\beta}, \frac{\delta^l}{\delta\ph^\al}\frac{\delta^r}{\ph^\beta}\right>\,.
\]
$\BV_\reg$ can be completed to a larger space, the space $\BV_\mc(M)$, of functionals that satisfy the condition
\be\label{mlsc}
\WF(F^{(n)}(\ph,\ph^\ddagger))\subset \Xi_n,\quad\forall n\in\NN,\ \forall\ph\in\overline{\E}(M)\,.
\ee
Functionals fulfilling this criterium are called \textit{microcausal}. On $\BV_\mc(M)$ we can define the star product:
\begin{equation*}%\label{star product}
F\star_H G\doteq m\circ \exp({i\hbar \Gamma'_\omega})(F\otimes G),
\end{equation*}
where  $\Gamma'_\omega$ is the functional differential operator
\begin{equation*}%\label{star product2}
\Gamma'_\omega\doteq  \sum_{\al, \beta} \left<{\omega}^{\al\beta},\frac{\delta^l}{\delta\ph^\al} \otimes \frac{\delta^r}{\delta\ph^\beta}\right>\,.
\end{equation*}
The resulting algebra is denoted by $\fA_H(M)$.
As there is no preferred two-point function $\omega$, and hence no preferred $H$, we have to consider all of them simultaneously. The quantum algebra (which contains in particular the Wick polynomials) is an extension of the source space $\BV_\reg(M)$ with respect to the initial topology induced by the map $\al_H:\BV_\reg\rightarrow \BV_\mc[[\hbar]]$ (see \cite{BDF} for details). It is defined by extending  $\BV_\reg(M)$ with all elements of the form $\lim_{n\rightarrow \infty}\al_H^{-1}(F_n)$, where $(F_n)$ is a convergent sequence in $\BV_\mc(M)$ with respect to $\tau_\Xi$. The resulting space, denoted by $\al_H^{-1}(\BV_\mc)$, is equipped with a unique continuous star product $\star$ equivalent to $\star_H$:
\[
\al_H^{-1}F\star \al_H^{-1}G\doteq \al_H^{-1}(F\star_H G)\,.
\]
Different choices of $H$ differ only by a smooth function, hence all the algebras  $(\al_H^{-1}(\BV_\mc[[\hbar]]),\star)$ are isomorphic and define an abstract algebra $\fA$. 
%Since for $F \in \fA$ it holds $\al_HF\in \fA_H$, 
We can realize $\fA$ more concretely as the space of families $\{ \al_HF \}_H$,  numbered by possible choices of $H$, fulfilling the relation
\[
 F_{H'} = \exp(\hbar \Gamma_{H'-H}) F_H\qquad F \in \fA\,,
\]
equipped with the product
\[
 (F \star G)_H = F_H \star_H G_H.
\]
The support of $F \in \fA(M)$ is defined as $\supp(F) = \supp(\al_HF)$. Again, this is independent of $H$. Functional derivatives are defined by
\be\label{derivH}
\skal{\frac{\delta F}{\delta \ph}}{\psi} = \al_H^{-1}\skal{\frac{\delta \al_HF}{\delta \ph}}{\psi}\,,
\ee
which is well defined as $\Gamma_{H'-H}$ commutes with functional derivatives. For a fixed background $M$, the free net of local algebras is defined by assigning to each relatively compact, causally convex region $\Ocal\subset(M)$, a unital $*$-algebra $\fA(\Ocal)$. 
%Polynomial functionals in $\fA_H(M)$ can be interpreted as Wick polynomials.
%% in the following sense:
%%\be\label{polynomials1}
%%\int \Phi_{x_1}\dots\Phi_{x_n} f(x_1,\dots, x_n)=:\int :\Phi_{x_1}\dots\Phi_{x_n}:_{\omega} f(x_1,\dots, x_n)\,,
%%\ee
%Corresponding elements of $\fA(M)$ can be obtained by applying $\al_H^{-1}$. The resulting object will be denoted by
%\be\label{polynomials1}
%\int :\Phi_{x_1}\dots\Phi_{x_n}:_H f(x_1,\dots, x_n)\doteq \al^{-1}_H\Big(\int \Phi_{x_1}\dots\Phi_{x_n} f(x_1,\dots, x_n)\Big)\,.
%\ee
%where $\Phi_{x_i}$ are evaluation functionals, $f\in\Ecal'_{\Xi_n}(M^n,V)$ and we suppress al the indices. 
\subsection{Interacting theory}\label{int:theor}
Following \cite{FR3}, we introduce the interaction by means of renormalized time-ordered products. 
 Let us define operators $\Tcal_{n}:\BV_\loc^{\otimes n}\rightarrow \fA$ by means of
\[
\Tcal_{n}(F_1,\ldots,F_n)=\al_{H+w}^{\minus} (F_1)\T\ldots\T \al_{H+w}^{\minus} (F_n)\,,
\]
for $F_i\in \BV_\loc$ with disjoint supports,  where $F\T G\doteq \alpha_{i\De_D}(\alpha_{i\De_D}^{\minus}F\cdot\alpha_{i\De_D}^{\minus}G)$, $\TT_{0}=0$ and $\TT_1=\al^{-1}_{H+w}$. In the last formula, $w$ is the smooth part in the Hadamard function. 
%$H=\frac{u}{\sigma}+v\ln\sigma+w$, where $\sigma(x,y)$ denotes the square of the length of the geodesic connecting $x$ and $y$ and  $u$ and $v$ are geometrically determined smooth functions. 
Renormalization freedom related to the choice of $w$ is discussed in \cite{HW}. Maps $\Tcal_{n}$ have to be extended to functionals with coinciding supports and are required to satisfy the standard conditions given in \cite{BDF,H}. In particular, we require the graded symmetry, unitarity, scaling properties, the support property $\supp\TT_n(F_1,\dots,F_n)\subset\bigcup\supp F_i$ and the causal factorization property, which states that
\be\label{CausFact}
\TT_{n}(F_1\otimes \dots \otimes F_n)=
\TT_{i}(F_1\otimes \dots \otimes F_i) \star
\TT_{n-i}(F_{i+1} \otimes \dots \otimes F_n) \, ,
\ee
if the supports of $F_1\ldots F_i$ are later than the supports of $F_{i+1},\ldots F_n$.

Maps $\TT_n$ are constructed inductively, and each $\TT_n$ is uniquely fixed by the lower order maps $\TT_k$, $k<n$, up to addition of an $n$-linear  map
\be
Z_n:\BV_\loc^n\to\al_{H+w}^{-1}(\BV_\loc)\cong\fA_\loc\,,
\ee
which describes the possible finite renormalizations.
In \cite{FR3} it was shown that renormalized time ordered product can be extended to an associative, commutative binary product defined on the domain $\Dcal_{\TT}\doteq\TT(\BV)$, where $\TT\doteq\oplus_n\TT_n\circ m^{-1}$.
Here $m^{-1}:\BV\to S^\bullet\BV^{(0)}_\loc$ is the inverse of the multiplication, as defined in \cite{FR3,Rej11b}. $\Dcal_{\TT}$ contains in particular $\fA_\loc$ and is invariant under the renormalization group action. Renormalized time ordered products are defined by
\be
A\T B\doteq\TT(\TT^{\minus}A\cdot\TT^{\minus}B)\,.
\ee
%For $n$ elements of $\fA_\loc(M)$ this amounts to
%\be\label{time-ordered-product-local}
%A_1\TRH\dots\TRH A_n=\TTR(\TT_1^{-1}A_1\cdot\ldots\cdot\TT_1^{-1}A_n)=\TT_n(\TT_1^{-1}A_1\cdot\ldots\cdot\TT_1^{-1}A_n)
%\ee
%Time ordered products on different spacetimes have to be defined in a covariant way. This means that if $\Phi\in\Fcle$ is a classical locally covariant field, then we have to define the families ${\TT_n}_M$
%in such a way that $(\TT\Phi)_M\doteq \TT_M\circ\Phi_M$ is an element of $\Fqe$.

Using time-ordered products we can introduce the interaction.
As indicated in section \ref{free:theory}, we split the Lagrangian into  $L_\ex=L_0+L_1$. The quantum field constructed form the interaction term $L_1$ is $\TT L_1$; the formal S-matrix is given by
\be\label{Smatrix}
\Scal(L_{1}(f))\doteq e_{\sst{\TT}}^{i\TT L_{1}(f)/\hbar}=\TT(e^{iL_{1}(f)/\hbar})\,,
\ee
which is a Laurent series in $\hbar$. Now we want to construct the interacting net of local algebras. Let $\Ocal\subset M$ be an open and relatively compact subset. 
The local algebra of observables associated to $\Ocal$ has to be independent of an interaction switched on outside
 of $\Ocal$.  We define
\[
\mathscr{V}_{S_1}(\Ocal) \doteq \{ V\in\fA_{\loc}\ |\ \supp(V-\TT{L_1}(f))\cap\overline{\Ocal}=\varnothing,
\text{ if } [L_1] =S_1\text{ and } f\equiv 1 \text{ on }\Ocal \}\,.
\]
The relative $S$-matrix in the algebraic adiabatic limit is given by
\[
\Scal^\Ocal_{S_1}(F)=(\Scal_V(F))_{V\in\mathscr{V}_{S_1}(\Ocal)}
\] 
for $F\in\fA$ with $\supp\, F\subset \Ocal$, where
\be\label{Bog}
\Scal_V(F)\doteq\Scal(V)^{\star-1}\star \Scal(V+F)\,.
\ee
The relative $S$-matrix defined this way is a covariantly constant section  in the sense that for any $V_1,V_2\in\mathscr{V}_{S_1}(\Ocal)$ there exists an automorphism
$\beta$ of $\fA$ such that
\[
\beta(\Scal_{V_1}(F))=\Scal_{V_2}(F)\quad\ \forall F\in\fA_{\loc}\ ,\ \supp\, F\subset \Ocal\,,
\] 
Interacting quantum fields in $\Ocal$ are generated by $\Scal_{V}^\Ocal(F)$ and, for given $V\in\mathscr{V}_{S_1}(\Ocal)$, we can write a corresponding component as a formal power series:
\be\label{Rv}
(R^\Ocal_{S_1}(F))_V\doteq R_V(F)=\frac{d}{d\lambda}\Big|_{\lambda=0}\Scal_{V}(\lambda F)\,,
\ee
Differentiation of $\Scal_{V}(\lambda F)$ yields
\be\label{RV}
R_V(F)=\left(e_{\sst{\TT}}^{iV/\hbar}\right)^{\star\minus}\star\left(e_{\sst{\TT}}^{iV/\hbar}\T F\right)\,,
\ee
which is a formal power series in $\hbar$ and $V$. An interacting net of local algebras on a fixed spacetime $M$ is obtained by assigning to each relatively compact, causally convex region $\Ocal\subset M$, an algebra $(R_V(\fA(\Ocal)),\star)$, $V\in\mathscr{V}_{S_1}(\Ocal)$. This definition doesn't depend on the choice of $V$, since local algebras constructed with different $V$'s belonging to $\mathscr{V}_{S_1}(\Ocal)$ are isomorphic.

An interacting field can be written in terms of \textit{retarded products} defined as coefficients in the following expansion:
\[
R_V(F)=\sum\limits_{n=0}^\infty\frac{i^n}{\hbar^nn!}\mathcal{R}_n(V^{\otimes n};F)\,.
\]
Time-ordered products can be normalized in such a way that retarded products satisfy a useful relation, called the GLZ identity,
\be\label{glz}
[R_V(F),R_V(G)]_\star=i\hbar\frac{d}{d\la}\left(R_{V+\la G}(F)-R_{V+\la F}(G)\right)\Big|_{\la=0}\,.
\ee
On the right hand side of \eqref{glz} we have formal derivatives of retarded products with respect to $V$, so it is convenient to use the notation
\[
R_V^{(1)}[G](F)\doteq\frac{d}{d\la}\left(R_{V+\la G}(F)\right)\Big|_{\la=0}\,.
\]
Retarded products satisfy an important support property, which can be conveniently written as
\be\label{supp:prop}
\supp(\mathcal{R}_n(V_1(x_1),\dots,V_1(x_n);F(x))\subset \{(x_1,\dots,x_n;x)|x_i\in x+\bar{V}_-, \forall i=1\dots,n\}\,,
\ee
where $F(x)$, $V_i(x_i)$  are local forms  (non-integrated, i.e. density-valued local functionals). Using retarded products we define a new non-commutative (partial) product $\star_V$. This product was first proposed by K.~Fredenhagen in \cite{F11} as an interacting star product. It is given by
\be\label{starV}
F\star_V G\doteq R_V^{-1}(R_V(F)\star R_V(G))\,,
\ee
for $F,G\in\fA$ such that this expression is well defined.

Classical structures appearing in the BV formalism also have to be quantized. The renormalized time-ordered antibracket is defined by
\[
\{X,Y\}_{\sst{\TT}}=\TT\{\TT^{-1}X,\TT^{-1}Y\}\ .
\]
We can also write it in the form
\be\label{antibracketTR}
\{X,Y\}_{\TT}=\sum_\al\int\!\left(\!\frac{\delta^r X}{\delta\ph^\al}\T\frac{\delta^l Y}{\delta\ph_\al^\dgr}-(-1)^{|\ph_\al^\dgr|}\frac{\delta^r X}{\delta\ph_\al^\dgr}\T\frac{\delta^l Y}{\delta\ph^\al}\!\right)d\mu\,,
\ee
where we denoted
\be\label{antibracketTR2}
\frac{\delta^r X}{\delta\ph^\al}\T\frac{\delta^l Y}{\delta\ph_\al^\dgr}\doteq \TT\Big(D^*\Big(\TT ^{-1}\frac{\delta X}{\delta\ph}\otimes \TT ^{-1}\frac{\delta Y}{\delta\ph^\ddagger}\Big)\Big)\,,
\ee
where $D^*$ is the pullback by the diagonal map and $\big(\TT ^{-1}\frac{\delta X}{\delta\ph}\big)(\ph)$ is a compactly supported distribution (i.e. an element of $\E'$) defined by
\[
\left<\big(\TT ^{-1}\frac{\delta X}{\delta\ph}\big)(\ph),f\right>\doteq\Big(\TT ^{-1}\Big<\frac{\delta X}{\delta\ph},f\Big>\Big)(\ph)=\Big<\frac{\delta}{\delta\ph}\TT ^{-1}X,f\Big>(\ph)\,,\qquad f\in\E\,.
\]
In the second step we used the field independence of time ordered products. Since $X\in\TT(\BV)$, the distribution $\big(\TT ^{-1}\frac{\delta F}{\delta\ph}\big)(\ph)$ defined by the above equation is actually an element of $\D$ and the pullback in \eqref{antibracketTR2} is well defined.
Similarly we define the antibracket with the $\star$-product:
\be\label{antibracketstar}
\{X,Y\}_{\star}=\sum_\al\int\!\left(\!\frac{\delta^r X}{\delta\ph^\al}\star\frac{\delta^l Y}{\delta\ph_\al^\dgr}-(-1)^{|\ph_\al^\dgr|}\frac{\delta^r X}{\delta\ph_\al^\dgr}\star\frac{\delta^l Y}{\delta\ph^\al}\!\right)d\mu\,,
\ee
whenever it exists. Clearly it is well defined if one of the arguments is regular or equal to $S_0$. Moreover, the antibracket $\{.,S_0\}_\star$ with the free action defines a $\star$-derivation. Similarly, $\{.,S_0\}_{\sst{\TT}}$ is a $\T$-derivation. A relation between, $\{.,S_0\}_{\sst{\TT}}$ and $\{.,S_0\}_\star$  is provided by the Master Ward Identity \cite{BreDue,H}:
\begin{align}\label{MWI}
\{e_{\sst{\TT}}^{iV/\hbar},S_0\}_\star&=\{ e_{\sst{\TT}}^{iV/\hbar},S_0\}_{\sst{\TT}}+e_{\sst{\TT}}^{iV/\hbar}\T(\Lap_V+\frac{i}{2\hbar}\{V,V\}_{\sst{\TT}})=\\
&=\frac{i}{\hbar}e_{\sst{\TT}}^{iV/\hbar}\T\big(\{V,S_0\}_{\sst{\TT}}+\frac{1}{2}\{V,V\}_{\sst{\TT}}-i\hbar\Lap(V)\big)\,,\nonumber
\end{align}
where $V\in\fA_\loc$ and $\Lap(V)$ is a local functional. One can expand it in powers of $V$ to obtain
\[
\Lap(V)=\sum_{n=0}^{\infty}\Lap^n(V^{\otimes n};V)\,,
\]
where coefficients $\Lap^n$ are linear functions $\fA_{\loc}^{\otimes n}\rightarrow \fA_{\loc}$ defined recursively by
\begin{align}
\Lap^n(V_1\otimes\ldots\otimes V_n; V)&=-\left(\tfrac{i}{\hbar}\right)^{n+1}V_1\T\ldots\T V_n\T\{V,S_0\}+\nonumber\\
&-\left(\tfrac{i}{\hbar}\right)^{n}\sum_{i=1}^nV_1\ldots\T \hat{V}_i\T\ldots V_n\T\int\frac{\delta V}{\delta\ph_\al^\ddagger(x)}\T\frac{\delta V_i}{\delta\ph^\al(x)}d\mu(x)+\nonumber\\
&-\sum_{I\subset \{1,\dots,n\}, I\neq \varnothing}\!\!\!\!\!\!\left(\tfrac{i}{\hbar}\right)^{|I|}\Tcirc_{i\in I}V_i\T\Lap^{|I^c|}\left(V^{\otimes(n-k)};V\right)+\nonumber\\
&+\left(\tfrac{i}{\hbar}\right)^{n+1}\int\left(V_1\T\dots\T V_n\T\frac{\delta V}{\delta\ph_\al^\ddagger(x)}\right)\star \frac{\delta S_0}{\delta\ph^\al(x)}d\mu(x)\,,\label{Lap:coeff}
\end{align}
where $\Tcirc_{i\in I}V_i$ denotes the $\T$-product of elements $V_i$ indexed by $i\in I$, and $I^c$ is the complement of $I$ in $\{1,\dots,n\}$. Let us now fix the interaction term $V$. Using the above relation we define the renormalized BV Laplacian on $\fA_\loc$ as: %\marginpar{rethink this definition}
\[
\Lap_V(X)\doteq \frac{d}{d\lambda}\Big|_{\lambda=0}\Lap({V+\lambda X})\,.
\]
Using this definition we obtain a relation:
\begin{align}\label{MWI2}
\Lap_V(X)=& \frac{d}{d\lambda}\Big|_{\lambda=0}(e_{\sst{\TT}}^{-i(V+\la X)/\hbar}\T(\{e_{\sst{\TT}}^{i(V+\la X)/\hbar},S_0\}_\star-\{ e_{\sst{\TT}}^{i(V+\la X)/\hbar},S_0\}_{\sst{\TT}})+\\&-\frac{i}{2\hbar}\{V+\la X,V+\la X\}_{\sst{\TT}})=\nonumber\\
=&\frac{i}{\hbar}(e_{\sst{\TT}}^{-iV/\hbar}\T\{e_{\sst{\TT}}^{iV/\hbar}\T X,S_0\}_\star- \{X,S_0+V\}+X\T e_{\sst{\TT}}^{-iV/\hbar}\T\{e_{\sst{\TT}}^{iV/\hbar},S_0\}_\star)\,.\nonumber
\end{align}
Expanding the renormalized Laplacian in powers of $V$ we obtain
\be\label{LapX:coeff}
\Lap_V(X)=\sum_{n=1}^\infty n\Lap^n\left(V^{\otimes (n-1)}\otimes X;V\right)+\sum_{n=0}^\infty\Lap^n(V^{\otimes n};X)\,.
\ee
If $X$ doesn't contain antifields, only the first sum is present. Compare $\Lap_V(X)$ with the nonrenormalized graded Laplacian $\Lap_{\textrm{nren}}$ of the BV formalism, which has all the $n>0$ terms vanishing and $\Lap_{\textrm{nren}}^0$ is given by the known formula:
\be\label{Lap:regular}
\Lap_{\textrm{nren}} X=\Lap_{\textrm{nren}}^0(X)=\sum\limits_\alpha(-1)^{|\ph_\al|(1+|X|)}\int dx \frac{\delta^2 X}{\delta\ph_\al^\ddagger(x)\delta\ph^\al(x)}\,.
\ee
Unfortunately $\Lap_{\textrm{nren}}$ is not well defined on local functionals. In the renormalized theory $\Lap_{\textrm{nren}}$ is replaced by $\Lap_V$, which is well defined on $\fA_\loc$, but contains non-vanishing higher order terms $\Lap^n$, $n>0$.

The renormalized quantum master equation QME is the condition that
\be\label{QME}
e_{\sst{\TT}}^{-i\TT S_1/\hbar}\T\left(\{e_{\sst{\TT}}^{i \TT S_1/\hbar},S_0\}_{\star}\right)\sim 0\,,
\ee
on the level of natural transformations. Using \eqref{MWI}, condition \eqref{QME} can be expressed as:
\[
\frac{1}{2}\{S_0+S_1,S_0+S_1\}-i\hbar\Lap({S_1})\sim 0\,,
\]
where $\Lap({S_1})$ is also seen as a natural transformation. Assume that the {\qme} holds for $S_1$ and let us fix 
 $V\in\mathscr{V}_{S_1}(\Ocal)$\footnote{The fulfillment of the {\qme} in the algebraic adiabatic limit is guaranteed by certain cohomological conditions, see \cite{FR3} and references therein. The problem is reduced to the analysis of the Lie algebra cohomology of the gauge algebra (the Lie algebra of the local symmetries Lie group). It is well known that such cohomological conditions are fulfilled, in particular, for  QED and Yang-Mills theories \cite{HennBar}, gravity \cite{BTM} and the bosonic string with the Nambu-Goto action\cite{BRZ}.} 
The quantum BV operator is defined by
\be\label{QBV}
\hat{s}(X)=e_{\sst{\TT}}^{-iV/\hbar}\T\left(\{e_{\sst{\TT}}^{iV/\hbar}\T  X,{S_0}(f)\}_{\star}\right)\,,
\ee
where $\supp\, X\subset\Ocal$ and $f\equiv 1$ on $\Ocal$ and it is independent of the choice of $V\in\mathscr{V}_{S_1}(\Ocal)$.
If (\ref{QME}) holds, then  it follows from \eqref{MWI2} that $\hat{s}$ can be expressed as
\be\label{class:quant}
\hat{s}(X)=\{X,S_0+V\}_{\sst\TT}-i\hbar\Lap_V(X)=sX-i\hbar\Lap_V(X)\,,
\ee
and it has the following property:
\be\label{intertwining:s:r}
\{.,S_0\}_\star\circ R_{S_1}^\Ocal=R_{S_1}^\Ocal\circ\hat{s}\,.
\ee
It was proven in \cite{FR3} that constructing a solution to (\ref{QME}) amounts to analyzing the cohomology $H^1(\gamma|d)$ on the space of local forms. The anomaly term $\Lap(V)$ is expressed in terms of the renormalized BV Laplacian with the use of fundamental theorem of calculus\footnote{The first version of \cite{FR3} contains a notational inconsistency which suggests that $\Lap(V)=\Lap_V(V)$. This was corrected in the erratum to that paper.}:
\[
\Lap(V)=\int\limits_0^1\Lap_{\lambda V}(V)d\lambda\,.
\]
The natural question to ask is: how to extend the operator $\Lap_V(.)$ to 
multilocal functionals? The nonrenormalized counterpart satisfies:
\be\label{Delta:Tbracket}
\Lap(X\T Y)=\Lap(X)\T Y+(-1)^{|X|}X\T \Lap(Y)+\{X,Y\}_{\TT}\,,
\ee
It would be tempting to require the same property to hold for the renormalized operator $\Lap_V(.)$, but then one would have to give up other properties. Note that since
\[
\{e_{\sst{\TT}}^{iV/\hbar}\T X\T Y,S_0\}_\star=-\hbar^2\frac{\partial^2}{\partial\lambda \partial\mu}\Big|_{\lambda=\mu=0}\{e_{\sst{\TT}}^{i(V+\lambda X+\mu Y)/\hbar},S_0\}_\star\,,
\]
one finds (assuming the QME and using \eqref{MWI})
\begin{multline}
\{e_{\sst{\TT}}^{iV/\hbar}\T X\T Y,S_0\}_\star=-\hbar^2\frac{\partial^2}{\partial\lambda \partial\mu}\Big|_{\lambda=\mu=0}\left(\{ e_{\sst{\TT}}^{i(V+\lambda X+\mu Y)/\hbar},S_0\}_{\sst{\TT}}+\right.\\
\left.e_{\sst{\TT}}^{i(V+\lambda X+\mu Y)/\hbar}\TR(\Lap({V+\lambda X+\mu Y})+\frac{i}{2\hbar}\{V+\lambda X+\mu Y,V+\lambda X+\mu Y\}_{\sst{\TT}})\right)=\\
e_{\sst{\TT}}^{iV}\T\Big(\{X\T Y,S_0+V\}-i\hbar\big(\Lap_V(X)\T Y+(-1)^{|X|}X\T \Lap_V(Y)+\{X,Y\}_{\TT}+\\
-i\hbar\frac{\partial^2}{\partial\lambda \partial\mu}\Big|_{\lambda=\mu=0}\Lap({V+\lambda X+\mu Y})\big)\Big)
\,.
\end{multline}
It follows that
\begin{multline*}
\hat{s}(X\T Y)=s(X\T Y)-i\hbar\Big(\Lap_V(X)\T Y+(-1)^{|X|}X\T \Lap_V(Y)+\{X,Y\}_{\TT}+\\-i\hbar\frac{\partial^2}{\partial\lambda \partial\mu}\Big|_{\lambda=\mu=0}\Lap({V+\lambda X+\mu Y})\Big)\,,
\end{multline*}
so in order to reproduce the relation \eqref{class:quant} also for products, it is natural to set
\begin{multline*}
\Lap_V(X\T Y)\doteq \Lap_V(X)\T Y+(-1)^{|X|}X\T \Lap_V(Y)+\{X,Y\}_{\TT}+\\-i\hbar\frac{\partial^2}{\partial\lambda \partial\mu}\Big|_{\lambda=\mu=0}\Lap({V+\lambda X+\mu Y})\,.
\end{multline*}
Note that on regular elements $\Lap_V(.)$ is just $\Lap_{\mathrm{nren}}$ and since it doesn't depend on the interaction, the last term in the above definition vanishes for $\Lap_{\mathrm{nren}}$, so our proposal is consistent with the non-renormalized case. Moreover, the ``extra term'' is of higher order in $\hbar$, so $\Lap_V$ behaves like the graded Laplacian, modulo $\Ocal(\hbar)$ corrections.
For higher powers we use an analogous definition (for simplicity of notation we assume all $X_i$ to be even)%\marginpar{some signs and maybe $\hbar$'s and $i$'s}:
\begin{multline}\label{DeltaV:sym}
\Lap_V(X_1\T\dots\T X_n)\doteq\\ \sum\limits_{I\subset\{1,...,n\}\atop I\neq \varnothing, |I|\neq n}\!\!\!(-i\hbar)^{|I_c|-1}\Tcirc_{i\in I} X_{i}\T\frac{\partial^{|I^c|}}{\partial\lambda_{j_1}\dots \partial\la_{j_{|I_c|}}}\Big|_{\vec{\la}_{I}=0}\Lap({V+\vec{\la}_I\cdot\vec{X}})+\\
+\sum_{i,j=1,\dots n\atop i< j}\{X_i,X_j\}_{\sst\TT}\T X_{1}\T\dots\widehat{X_i}\T\dots\widehat{X_j}\T\dots X_n\,,
\end{multline}
where $\vec{\la}_I=(\la_{j_1},...,\la_{j_{|I^c|}})$ and $\vec{\la}_I\cdot\vec{X}\doteq \sum_{j}\lambda_j X_j$.%\marginpar{What conditions are needed in order to make it into an $S_\infty$ algebra?}
In \cite{FR3} we showed the existence of a map  from multilocal functionals to the graded symmetric algebra over local functionals $\beta:\BV\rightarrow S^\bullet\BV_\loc$, which is the inverse of the multiplication $m:S^\bullet\BV_\loc\rightarrow\BV$. Using this map and the definition \eqref{DeltaV:sym} one can extend $\Lap_V(.)$ to $\TT(\BV)$ and, for all $X\in\TT(\BV)$,
\[
\hat{s}X=sX-i\hbar\Lap_VX
\]
holds. In the special case of the time-ordered exponential $e_{\sst \TT}^{iX/\hbar}$, $X\in\BV_\loc$, we obtain
\be\label{Lap:exp}
\Lap_V\big(e_{\sst \TT}^{iX/\hbar}\big)=e_{\sst \TT}^{iX/\hbar}\T\Big(\tfrac{i}{\hbar}(\Lap({V+X})-\Lap(V))+\left(\tfrac{i}{\hbar}\right)^2\tfrac{1}{2}\{X,X\}_{\sst \TT}\Big)\,.
\ee
Note that \eqref{intertwining:s:r} implies that $\hat{s}=R_V^{-1}\circ\{.,S_0\}_\star\circ R_V$, so from the nilpotency of 
$\{.,S_0\}_\star$ follows that $\hat{s}$ is also nilpotent. The algebra of gauge invariant quantum fields is defined as the cohomology of the quantum BV operator $\hat{s}$. We want to stress that the quantization scheme proposed by K.~Fredenhagen and myself in \cite{FR3} doesn't require $\gamma_0$ to be a derivation with respect to $\star$. This is due to the fact that $\theta_0$ is included in the interaction term $V$ and the free action $S_0$ is just the quadratic part of the gauge fixed action $S$. Therefore, our approach is suitable for theories, like gravity, where the invariance of the star product with respect to $\gamma_0$ is not easy to establish.
\subsection{Quantization of $S_0+\theta_0$}\label{gauge:inv}
In this section we discuss the possibility to consider  $S_0+\theta_0$, instead of $S_0$, as the starting point for our construction. In order to do it, we have to ensure that the $\star$-antibracket with $\theta_0$ is well defined on multilocal functionals and that $\{.,\theta_0\}_\star$ is a $\star$-derivation. We will show that this is possible only if we require additional conditions, related to the free {\cme}.
%The next calculation shows how one writes $\{\theta_0,L_0\}$ in terms of operators $P$ and $K$.
%\begin{align}
%\{\theta_0,L_0\}(f_1,f_2)=&-\int (K^\gamma_{\ \sigma}\Phi^\sigma_z)(f(x)\delta(x-z)P_{\gamma\beta}(x)f(x)\Phi^\beta_x+\nonumber\\
%&+(-1)^{|\ph^\beta||\ph^\gamma|}f(x)\Phi^\beta_xP_{\beta\gamma}(x)f(x)\delta(x-z)) d\mu(x,z)=\nonumber\\
%=&-2\int f(z)(K^\gamma_{\ \sigma}\Phi^\sigma_z)P_{\gamma\beta}(f\Phi^\beta_z) d\mu(z)=2\left<fK\Phi,P(f\Phi)\right>\,,\label{CME2}
%\end{align}
%%From the CME of the free theory one obtains a following identity for the operators $K$ and $P$:
%%\be\label{KPrel}
%%(K^\gamma_{\ \sigma}\Phi^\sigma_z)P_{\gamma\beta}\Phi^\beta_z=dF\,,
%%\ee
%where  $f_2\equiv 1$ on the support of $f_1$.

We start our discussion with a slight reformulation of condition  \eqref{freeCME}. Let us fix a compact region $\mathcal{K}\subset M$. From \eqref{freeCME} follows that the relation $\left<K\psi,P\psi\right>_{\overline{\E}}=0$
%$\int (K^\gamma_{\ \sigma}\psi^\sigma_z)P_{\gamma\beta}(\psi^\beta_z) d\mu(z)=0$
 holds for all $\psi\in\overline{\E}_c$ with $\supp(\psi)\subset \mathcal{K}$. This can be also be written as 
 \be\label{PKKP}
\left<P^*K\psi,\psi\right>_{\overline{\E}}+\left<\psi,K^*P\psi\right>_{\overline{\E}}\,,
 \ee
where $*$ denotes the formal adjoint with respect to $\left<.,.\right>_{\overline{\E}}$. Let us define for any operator $O$ on $\overline{\E}$, the following operation:  $(O^\dagger)^{\al}_{\ \beta}=(-1)^{|\ph^\al||\ph^\beta|}(O^{\ \al}_{\beta})^*$ (the graded formal adjoint). Note that $P^\dagger=P$ and  from \eqref{PKKP} we see that \eqref{freeCME} is satisfied, if the following, stronger, condition is fulfilled:
\be\label{PK}
(-1)^{|\ph^\beta|}P_{\beta\gamma}K^\gamma_{\ \sigma}+(K^\dagger)^{\ \gamma}_{\beta} P_{\gamma\sigma}=0\,.
\ee
For determining the signs we used the fact that $\gamma_0$ is an odd differential, so $K^\gamma_{\ \al}\neq 0$ only for $|\ph^\al|+|\ph^\gamma|=1\mod2$. The same holds for $PK$, because $S_0$ is even. 

Comparing with \eqref{fullCME}, we conclude that \eqref{PK} is compatible with the full {\cme} only if $\ph_0$ is on-shell. We will now study in detail the consequences of condition \eqref{PK}.
\begin{prop}
Identity \eqref{PK} implies that the linearized Koszul-Tate operator $\delta_0$ anticommutes with $\gamma_0$:
\[
 \gamma_0\circ\delta_0+\delta_0\circ\gamma_0=0\,.
\]
\begin{proof}
To see this, note that for a constant derivation $X=\int X^\al(x)\frac{\delta}{\delta\ph^\al(x)}d\mu(x)$ we have
\begin{multline*}
\gamma_0\circ\delta_0X=-\gamma_0\int X^\al(x)\frac{\delta_l L_0(f)}{\delta\ph^\al(x)}d\mu(x)\Big|_{f\equiv 1\atop\textrm{ on }\supp X}=\\
=\int (-1)^{|\ph^\ga|}X^\al(x)({K}^{\gamma}_{\ \sigma}\Phi_z^\sigma) P_{\al\gamma}(x)\delta(x-z)d\mu(x,z)\,.
\end{multline*}
On the other hand
\begin{align*}
\delta_0\circ\gamma_0X&=\int X^\al(x) K^\gamma_{\ \al}(x)^*\frac{\delta L_0(f)}{\delta\ph^\gamma(x)}d\mu(x)\Big|_{f\equiv 1\atop\textrm{ on }\supp X}=\\
&=\int X^\al(x) K^\gamma_{\ \al}(x)^*P_{\gamma\sigma}(x)\Phi^\sigma_xd\mu(x)\,.
\end{align*}
Therefore
\[
(\gamma_0\circ\delta_0+\delta_0\circ\gamma_0)X=\int X^\al(x)((-1)^{|\ph^\al|} P_{\al\gamma}{K}^{\gamma}_{\ \sigma} + {K^\dagger}^{\ \gamma}_{\al}P_{\gamma\sigma})\Phi_x^\sigma d\mu(x)=0
\]
follows.\end{proof}\end{prop}
\begin{exa}[Free electromagnetic field]
Using results from Example \ref{ex1}, we can verify that the condition \eqref{PK} holds for the free electromagnetic field. Note that
\[
K^\dagger=\left(\begin{array}{cccc}
0&0&0&0\\
0&0&0&-i\\
-\delta&0&0&0\\
0&0&0&0
\end{array}\right)\,,
\]
and the direct computation shows that
\begin{multline*}
(-1)^{|\ph^\beta|}P_{\beta\gamma}K^\gamma_{\ \sigma}+(K^\dagger)^{\ \gamma}_{\beta} P_{\gamma\sigma}=\\
\left(\begin{array}{cccc}
0&0&0&0\\
0&0&\delta d&0\\
0&\delta d&0&0\\
0&0&0&0
\end{array}\right)-\left(\begin{array}{cccc}
0&0&0&0\\
0&0&\delta d&0\\
0&\delta d&0&0\\
0&0&0&0
\end{array}\right)=0
\end{multline*}\,.
\end{exa}
Relation \eqref{PK} allows us to prove the so called \textit{consistency conditions}\footnote{I would like to thank Jochen Zahn for enlightening discussions about the importance of consistency conditions and for crucial remarks on the proof of proposition \ref{gauge:inv:Delta}.}, formulated first in \cite{H} in the case of Yang Mills theory and generalized in \cite{BRZ}. 
\begin{prop}\label{gauge:inv:Delta}
Assume that \eqref{PK} holds and that $S_0$ induces a normally hyperbolic system of equations of motion: $P\ph=0$. Let $\Delta^*$ be the retarded, the advanced or the causal propagator corresponding to $P$. Then $\Delta^*$ satisfies the consistency conditions:
\be\label{const:cond}
\sum_\sigma((-1)^{|\ph^\al|}K^{\al}_{\ \sigma}(x')\Delta^*(x',x)^{\sigma\gamma}+K^{\gamma}_{\ \sigma}(x)\Delta^*(x',x)^{\al\sigma})=0\,.
\ee
\begin{proof} First, we prove the property \eqref{const:cond} for $\Delta^R$.  We act with $\eqref{PK}$ on $\Delta^R$, which yields
\[
(-1)^{|\ph^\al|}P_{\al\gamma}(z)\circ K^\gamma_{\ \sigma}(z)\Delta^R(z,x)^{\sigma\beta}=-K^\gamma_{\ \al}(z)^*\delta(z-x)\delta_\gamma^\beta\,.
\]
We multiply both sides of the above identity with $\Delta^A(z,y)^{\al\mu}$. The integration over $z$ results in
\[
(-1)^{|\ph^\al|}\int\Delta^A(z,y)^{\al\mu}P_{\al\gamma}(z)\circ K^\gamma_{\ \sigma}(z)\Delta^R(z,x)^{\sigma\beta}d\mu(z)=-K^\beta_{\ \al}(x)\Delta^A(x,y)^{\al\mu}\,.
\]
Next, we use the integration by parts to make $P$ act on $\Delta^A$ from the left. This is possible, since (due to support properties of $\Delta^A$ and $\Delta^R$) the integrant is supported in the intersection of the past of $y$ and the future of $x$, which is a compact set. We obtain
 \[
K^\mu_{\ \sigma}(y)\Delta^R(y,x)^{\sigma\beta}=-K^\beta_{\ \al}(x)\Delta^A(x,y)^{\al\mu}\,.
\]
Using the relation between the retarded and advanced propagators we rewrite the above expression as
 \[
(-1)^{|\ph^\al|}K^\mu_{\ \sigma}(y)\Delta^R(y,x)^{\sigma\beta}+K^\beta_{\ \al}(x)\Delta^R(y,x)^{\mu\al}\,.
\]
The same follows for $\Delta^A$ and also for the difference of the two.
\end{proof}
\end{prop}
\begin{exa}[Free electromagnetic field]
For the free electromagnetic field, the retarded and the advanced propagators take the form:
\[
\De^{R/A}=\left(\begin{array}{cccc}
\De_v^{R/A}&d\De_v^{R/A}&0&0\\
\delta\De_v^{R/A}&-d\delta\De_v^{R/A}&0&-i\\
-\delta&0&0&i\De_s^{R/A}\\
0&0&-i\De_s^{R/A}&0
\end{array}\right)\,,
\]
where $\De_v^{R/A}$ are the propagators corresponding to the Laplace operator $\de d+d\de$ acting on 1-forms and  $\De_s^{R/A}$ are the propagators of $\de d$ acting on 0-forms. Property \eqref{const:cond} is expressed as
\[
d_x\De_s(x,y)+\delta_y\De_v(x,y)=0\,.
\]
\end{exa}

We will show now that \eqref{const:cond} is a necessary condition for $\{.\theta_0\}_\star$ to be well defined on local functionals. Recall that the  $\star$-antibracket with $\theta_0$ is given by
%\be\label{antibracketstar2}
%\{X,\theta_0\}_{\star}=\left<D^*\sum_\al\!\left(\!\frac{\delta^r X}{\delta\ph^\al}\star\frac{\delta^l \theta_0(f)}{\delta\ph_\al^\dgr}-(-1)^{|\ph_\al^\dgr|}\frac{\delta^r X}{\delta\ph_\al^\dgr}\star\frac{\delta^l \theta_0(f)}{\delta\ph^\al}\!\right),1\right>\,,
%\ee
\be\label{star:br:th0}
\{.,\theta_0(f)\}_\star=\{.,\theta_0(f)\}+\int\! d\mu(x,y,x') \Delta(y,x)^{\sigma\gamma}\frac{{\delta^l}^2 \theta_0(f)}{\delta\ph_\sigma(y)\delta\ph_\al^\dgr(x')}\frac{{\delta^r}^2}{\delta\ph^\gamma(x)\delta\ph^\al(x')}\,,
\ee
Since $\theta_0$ is local, the second term is not well defined for local arguments, hence we require that it vanishes identically.  Note that we can write this term as
\begin{multline}
\int\! d\mu(x,y,x') f(y)\Delta(y,x)^{\sigma\gamma}(-1)^{|\ph^\al|+1}K^{\al}_{\phantom{\al}\sigma}(x')\delta(y-x')\frac{{\delta^r}^2}{\delta\ph^\gamma(x)\delta\ph^\al(x')}=\\=\int\! d\mu(x,x') (-1)^{|\ph^\al|+1}K^{\al}_{\phantom{\al}\sigma}(x')(f(x')\Delta(x',x)^{\sigma\gamma})\frac{{\delta^r}^2}{\delta\ph^\gamma(x)\delta\ph^\al(x')}\,.
\end{multline}
We can express the above formula as a sum of two terms and in the second term we rename the indices $\al$ and $\gamma$. We obtain
\begin{multline*}
\int\! d\mu(x,x')\left( (-1)^{|\ph^\al|+1}K^{\al}_{\phantom{\al}\sigma}(x')(f(x')\Delta(x',x)^{\sigma\gamma})\frac{{\delta^r}^2}{\delta\ph^\gamma(x)\delta\ph^\al(x')}\right.+\\
+\left.(-1)^{|\ph^\gamma|+1}K^{\gamma}_{\phantom{\gamma}\sigma}(x')(f(x')\Delta(x',x)^{\sigma\al})\frac{{\delta^r}^2}{\delta\ph^\al(x)\delta\ph^\gamma(x')}\right)\,.
\end{multline*}
Next we use the graded antisymmetry of $\Delta$ and in the second term we rename the integration variables $x$ and $x'$. This results in
\begin{multline}\label{second:term}
\int\! d\mu(x,x')\left( (-1)^{|\ph^\al|+1}K^{\al}_{\phantom{\al}\sigma}(x')(f(x')\Delta(x',x)^{\sigma\gamma})\right.+\\
+\left.K^{\gamma}_{\phantom{\gamma}\sigma}(x)(f(x)\Delta(x',x)^{\al\sigma})\right)\frac{{\delta^r}^2}{\delta\ph^\al(x')\delta\ph^\gamma(x)}\,.
\end{multline}
Now we note that in the definition of $\{X,\theta_0\}_\star$, $f$ has to be chosen to be identically 1 on the support of $X$, so \eqref{second:term} vanishes due to identity \eqref{const:cond} of lemma \ref{gauge:inv:Delta} and the following identities hold:
\[
\{.,\theta_0\}_\star=\{.,\theta_0\}=\gamma_0\,.
\]
We can say that the classical BRST symmetry survives in the free quantized theory.
Next we have to check if $\gamma_0$ is a derivation with respect to the $\star$-product. This is done in the following proposition. 
\begin{prop}\label{theta:derivation}
Let $S_0$ be the quadratic term of the action with $\#\af=0$ and let $\gamma_0$ be the free BRST operator. 
Assume that \eqref{const:cond} holds. Then, for $X,Y\in\BV_{\reg}$:
\be\label{theta:deriv0}
\gamma_0(X\star Y)=\gamma_0X\star Y+(-1)^{|X|}X\star\gamma_0Y\,.
\ee
\begin{proof}
Recall that $X\star Y\doteq m\circ \exp({i\hbar \Gamma'_\Delta})(X\otimes Y)$. For simplicity we write the proof for $X$ and $Y$ even. The general case differs by introducing some additional signs. From the graded Leibniz rule follows that
\[
\gamma_0\circ m (X\otimes Y)=m\circ(\gamma_0\otimes 1+1\otimes\gamma_0)(X\otimes Y)\,.
\]
Clearly $\gamma_0$ is a derivation if
\[
m\circ\big((\gamma_0\otimes 1+1\otimes\gamma_0)\circ\exp({i\hbar \Gamma'_\Delta})\big)=m\circ\big(\exp({i\hbar \Gamma'_\Delta})\circ(\gamma_0\otimes 1+1\otimes\gamma_0)\big)
\]
holds. Inserting \eqref{gamma0} we obtain the condition
\begin{multline*}
\int \Delta^{\al\beta}(x,y){K}^{\gamma}_{\ \sigma}(z)\delta(x-z)\delta_{\al\sigma}\frac{\delta_l}{\delta\ph^\gamma(z)}\otimes\frac{\delta_l}{\delta\ph^\beta(y)}d\mu(x,y,z)+\\+\int \Delta^{\al\beta}(x,y){K}^{\gamma}_{\ \sigma}(z)\delta(y-z)\delta_{\beta\sigma}\frac{\delta_r}{\delta\ph^\al(x)}\otimes\frac{\delta_l}{\delta\ph^\gamma(z)}d\mu(x,y,z)=0\,.
\end{multline*}
Next we change one of the derivatives in the first term from a left to a right one and we perform the integrations over the delta distributions. The above condition becomes:
\begin{multline*}
\int {K}^{\gamma}_{\ \beta}(x)\Delta^{\beta\al}(x,y)(-1)^{|\ph^\gamma|}\frac{\delta_r}{\delta\ph^\gamma(x)}\otimes\frac{\delta_l}{\delta\ph^\al(y)}d\mu(x,y)+\\+\int {K}^{\gamma}_{\ \beta}(y)\Delta^{\al\beta}(x,y)\frac{\delta_r}{\delta\ph^\al(x)}\otimes\frac{\delta_l}{\delta\ph^\gamma(y)}d\mu(x,y)=0\,.
\end{multline*}
Renaming the summation indices we see that the above condition is fulfilled, if $\Delta$ satisfies:
\[
(-1)^{|\ph^\gamma|}{K}^{\gamma}_{\ \beta}(x)\Delta^{\beta\al}(x,y)+{K}^{\al}_{\ \beta}(y)\Delta^{\gamma\beta}(x,y)=0\,,
\]
which is exactly \eqref{const:cond}.
\end{proof}\end{prop}
%This agrees with the so called ``consistency conditions'' imposed in \cite{H} for the Hadamard function in case of Yang Mills theory and generalized in \cite{BRZ}. Note that in our case this condition has to hold only for $\Delta$, because arguments of $\{.,\theta_0\}_\star$ are elements of the abstract algebra $\fA$, where the star product is $\star$. 
We have seen that consistency conditions \eqref{const:cond} are necessary for the BRST construction in the free theory and that they follow automatically for the causal propagator $\Delta$. Now let $\omega=\tfrac{i}{2}\Delta+H$ be a 2-point function  of some quasifree Hadamard state. We introduce a following definition
\begin{df}
A Hadamard 2-point function $\omega$ is said to be 
\begin{enumerate}
\item \textbf{gauge invariant} if it satisfies the condition analogous to the one fulfilled by $\Delta$:
\be\label{const:cond2}
\sum_\sigma((-1)^{|\ph^\al|}K^{\al}_{\ \sigma}(x')\omega(x',x)^{\sigma\gamma}+K^{\gamma}_{\ \sigma}(x)\omega(x',x)^{\al\sigma})=0\,,
\ee
\item \textbf{gauge invariant modulo a smooth function} if it satisfies
\be\label{const:cond3}
\sum_\sigma((-1)^{|\ph^\al|}K^{\al}_{\ \sigma}(x')\omega(x',x)^{\sigma\gamma}+K^{\gamma}_{\ \sigma}(x)\omega(x',x)^{\al\sigma})=0\ \textrm{mod }\Ci\textrm{ function}\,.
\ee
\end{enumerate}
\end{df}
Let us now explain in detail why the consistency conditions are needed. Note that
\[
\{X,Y\}_{\star_H}=\al_H\{\al_H^{-1}X,\al_H^{-1}Y\}_{\star}\,,
\]
and since $\al_H^{-1}\theta_0(f)=\theta_0(f)$, we obtain
\[
\{X,\theta_0(f)\}_{\star_H}=\al_H\{\al_H^{-1}X,\theta_0(f)\}_{\star}
\]
for regular $X$. If \eqref{PK} holds, we have
\[
\{X,\theta_0(f)\}_{\star_H}=\al_H\{\al_H^{-1}X,\theta_0(f)\}=(\al_H\circ\gamma_0\circ\al_H^{-1})(X)\,.
\,,
\]
Let us denote $\gamma_0^H\doteq\al_H\circ \gamma_0\circ\al_H^{-1}$. Since
\[
\{\al_H^{-1}X,\theta_0\}_{\star}=\gamma_0(\al_H^{-1}X)=\al_H^{-1}\circ \gamma_0^H X\,,
\]
we can interpret $\gamma_0$ on $\fA$ as the normal ordered counterpart of $\gamma_0^H$ on $\fA_H$. 
Now, let us take an arbitrary (not necessarly regular) $F\in\fA_H(\Ocal)$ and express it as a limit of the series of regular functionals $F=\lim_{n\rightarrow\infty}F_n$. Since $F_n$'s are regular,
\be\label{gamma:H}
\gamma_0^HF_n=\{F_n,\theta_0\}+\int\frac{\delta^2F_n}{\delta\ph^\al(x)\delta\ph^\beta(y)} H^{\beta\gamma}(y,z)\frac{\delta^2\theta_0}{\delta\ph^\gamma(z)\delta\ph^\dag_\al(x)}d\mu(x,y,z)
\ee
holds. The second term in the above expression is not well defined for local $F$, due to singularities of $H$. For $\gamma_0^H$ to be well defined on the full space $\fA_H(M)$, we have to require that $\int {H}^{\beta\gamma}(y,z)\frac{\delta^2\theta_0}{\delta\ph^\al(z)\delta\ph^\dag_\al(x)}d\mu(z)$ vanishes in the coinciding point limit $x\rightarrow y$, modulo a smooth function. Using the graded symmetry of the second derivative, we find that this requirement is
 equivalent to the condition \eqref{const:cond3}. Let us assume that there exists at least one $H$ for which $\omega=\frac{i}{2}\Delta+H$ fulfills \eqref{const:cond3}.
Then, if we take an arbitrary parametrix $\omega'=\frac{i}{2}\Delta+H'$, such that $H-H'$ is smooth, expression \eqref{gamma:H} has a well defined limit as well and $\gamma_0^{H'}F$ is well defined for all $F\in\fA_{H'}(\Ocal)$.

We have seen that the existence of $\gamma_0^H$ on $\fA_H(\Ocal)$ requires the condition \eqref{const:cond3} to be fulfilled. One can reach exactly the same conclusion working directly with $\gamma_0$ on $\fA(\Ocal)$. By the definition of the initial topology on $\fA(\Ocal)$ we know that the limit of $\gamma_0(\al_H^{-1}F_n)$ exists as an element of $\fA(\Ocal)$ if there exists an $H'$ such that $(\al_{H'}\circ\gamma_0\circ\al_H^{-1})(F_n)$ converges in $\fA_{H'}(\Ocal)$. Let us write this expression in a different way:
\begin{multline*}
(\al_{H'}\circ\gamma_0\circ\al_H^{-1})(F_n)=(\al_{H'-H}\circ\gamma^H_0)(F_n)=\\=\{\al_{H'-H}F_n,\theta_0\}+\int\frac{\delta^2(\al_{H'-H}F_n)}{\delta\ph^\al(x)\delta\ph^\al(y)} {H'}^{\beta\gamma}(y,z)\frac{\delta^2\theta_0}{\delta\ph^\al(z)\delta\ph^\dag_\al(x)}d\mu(x,y,z)\,.
\end{multline*}
If \eqref{const:cond2} is fulfilled, then the second term in the above expression vanishes and  $(\al_{H'}\circ\gamma_0\circ\al_H^{-1})(F_n)=\{\al_{H'-H}F_n,\theta_0\}$ converges to a microcausal functional in $\fA_{H'}(\Ocal)$, so $\lim_{n\rightarrow\infty}\gamma_0(\al_H^{-1}F_n)$ is a well defined element of $\fA(\Ocal)$.

Let us now discuss another possibility to define the BRST operator on $\fA_H(\Ocal)$ by using $\gamma_0$ instead of $\gamma_0^H$. In other words, we subtract ``by hand'' the singular term in \eqref{gamma:H}. For this to work, one has to prove that $\gamma_0$ is a derivation on $\fA_H(\Ocal)$. Here the consistency condition \eqref{const:cond2} enters again. We replace $\Delta$ by $\omega$ in theorem \ref{theta:derivation} and conclude that a sufficient condition for $\gamma_0$ to be a derivation with respect to $\star_H$ is \eqref{const:cond2}, which is equivalent to the requirement that $\gamma_0$ commutes with $\al_H$, i.e. $\al_H$ induces a cochain morphism. In general, this seems to be too strong, since we want to work with $\omega$ which is a parametrix but not a bisolution and one expects that \eqref{const:cond3} rather than \eqref{const:cond2} holds. Therefore it is more natural to work with $\gamma_0^H$ instead of $\gamma_0$.

To summarize, consistent BV quantization of the free theory can be performed if we can show the existence of at least one quasifree Hadamard state with a 2-point function satisfying \eqref{const:cond2}. This problem has not yet been solved in full generality\footnote{For Yang-Mills theory with a trivial principal bundle and $0$ background section one can use a deformation argument of \cite{FNW}, as it was done in \cite{H}. Problems start, however, if one allows nontrivial topology of principal bundles of the theory and considers arbitrary background connections. This issue is currently investigated by Jochen Zahn. Another example of a theory where the existence of a gauge invariant  Hadamard 2-point function is not clear is perturbative quantum gravity \cite{BFR}. }. Therefore, we think that it is more convenient to use $S_0$ as the free action, since this choice doesn't require any additional conditions. 
\subsection{Changing the free theory}\label{changing}
In section \ref{gauge:inv} we have shown that the {\cme} of the free theory is a necessary condition that allows us to construct the free quantum theory corresponding to action $S_0+\theta_0$. Now we want to include the interaction into the discussion. It was proven in \cite{FR3} by K.~Fredenhagen and myself that the quantum master equation is a necessary condition for the gauge invariance of the interacting theory. There, we considered the perturbation around the free action $S_0$. In this section, we  show that  full {\qme} can be equivalently formulated for $S_0+\theta_0$, provided the {\qme} of the free theory holds. For the beginning, we consider only the regular functions $\BV_{\reg}$ and the non-renormalized time ordered product. Let $V,\, \theta_0\in\BV_{\reg}$ and denote $\tilde{V}\doteq V-\theta_0$. The {\qme} is the condition that:
\be\label{QME:nonren}
e_{\sst{\TT}}^{-i(\tilde{V}+\theta_0)/\hbar}\T\left(\{e_{\sst{\TT}}^{i (\tilde{V}+\theta_0)/\hbar},S_0\}_{\star}\right)=0\,.
 \ee
Using properties of $\T$ and $\star$ we can rewrite this condition as:
 \begin{multline*}
e_{\sst{\TT}}^{-i(\tilde{V}+\theta_0)/\hbar}\T\left(\{e_{\sst{\TT}}^{i (\tilde{V}+\theta_0)/\hbar},S_0\}_{\star}\right)=\\
e_{\sst{\TT}}^{-i(\tilde{V}+\theta_0)/\hbar}\T\left(\{e_{\sst{\TT}}^{i\tilde{V}/\hbar}\T e_{\sst{\TT}}^{i\theta_0/\hbar},S_0\}_{\sst{\TT}}+i\hbar\Lap_{\mathrm{nren}}(e_{\sst{\TT}}^{i\tilde{V}/\hbar}\T e_{\sst{\TT}}^{i\theta_0/\hbar})\right)=\\
e_{\sst{\TT}}^{-i\theta_0/\hbar}\T\left(\{e_{\sst{\TT}}^{i\theta_0/\hbar},S_0\}_{\sst{\TT}}+i\hbar\Lap_{\mathrm{nren}}(e_{\sst{\TT}}^{i\theta_0/\hbar})\right)+\\+e_{\sst{\TT}}^{-i\tilde{V}/\hbar}\T\left(\{e_{\sst{\TT}}^{i\tilde{V}/\hbar},S_0+\theta_0\}_{\sst{\TT}}+i\hbar\Lap_{\mathrm{nren}}( e_{\sst{\TT}}^{i\tilde{V}/\hbar})\right)=\\
e_{\sst{\TT}}^{-i\theta_0/\hbar}\T\left(\{e_{\sst{\TT}}^{i \theta_0/\hbar},S_0\}_{\star}\right)+e_{\sst{\TT}}^{-i\tilde{V}/\hbar}\T\left(\{e_{\sst{\TT}}^{i\tilde{V}/\hbar},S_0+\theta_0\}_{\sst{\TT}}+i\hbar\Lap_{\mathrm{nren}}( e_{\sst{\TT}}^{i\tilde{V}/\hbar})\right)
\,.
 \end{multline*}
If the {\qme} of the free theory holds, i.e. if $e_{\sst{\TT}}^{-i\theta_0/\hbar}\T\left(\{e_{\sst{\TT}}^{i \theta_0/\hbar},S_0\}_{\star}\right)=0$, then \eqref{QME:nonren} is equivalent to:
 \begin{multline}\label{QME:modified}
 e_{\sst{\TT}}^{-i\tilde{V}/\hbar}\T\left(\{e_{\sst{\TT}}^{i\tilde{V}/\hbar},S_0+\theta_0\}_{\sst{\TT}}+i\hbar\Lap_{\mathrm{nren}}( e_{\sst{\TT}}^{i\tilde{V}/\hbar})\right)=\\= e_{\sst{\TT}}^{-i\tilde{V}/\hbar}\T\left(\{e_{\sst{\TT}}^{i\tilde{V}/\hbar},S_0\}_{\star}+\{e_{\sst{\TT}}^{i\tilde{V}/\hbar},\theta_0\}_{\sst{\TT}}\right)=0\,.
 \end{multline}
It was proven in section \ref{gauge:inv} that the free {\cme} implies $\{e_{\sst{\TT}}^{i\tilde{V}/\hbar},\theta_0\}_{\sst{\TT}}=\{e_{\sst{\TT}}^{i\tilde{V}/\hbar},\theta_0\}_{\star}$, so finally we can write \eqref{QME:modified} as:
\[
 e_{\sst{\TT}}^{-i\tilde{V}/\hbar}\T\left(\{e_{\sst{\TT}}^{i\tilde{V}/\hbar},S_0+\theta_0\}_{\star}\right)=0\,.
\]
After this short introduction we can come back to the discussion of the renormalized time-ordered product. To distinguish it from the non-renormalized one, we denote it in this subsection by $\TR$. 

First we have to check if the free {\qme} can be satisfied by exploiting the renormalization freedom which we have in defining  $\TR$. 
Note that, since $\theta_0$ is linear in both fields and antifields and is assumed to be of degree 0, the component $\theta^\alpha(x)$ doesn't depend on field $\ph^\alpha(x)$. It follows that $\{\theta_{0},\theta_0\}=0$ and using the anomalous {\mwi} \eqref{MWI} we obtain:
\[
e_{\sst{\TTR}}^{-i\theta_{0}/\hbar}\TR\{e_{\sst{\TTR}}^{i\theta_{0}/\hbar},S_0\}_\star=\{\theta_{0},S_0\}_{\sst{\TTR}}+i\hbar\Lap({\theta_{0}})\,,
\]
where $\Lap({\theta_0})$ is the anomaly term. Using the standard arguments of \cite{H,FR3} we can conclude that it must be constructed from elements of the relative cohomology $H^1(\gamma_0|d)$ on the space of local forms. If this cohomology is trivial, then the free theory is anomaly free, i.e. we can use the renormlization freedom to redefine the time-ordered powers of  $\theta_0$ to obtain $\Lap({\theta_0})=0$, so the free {\qme} holds as a consequence of free {\cme}.

Now, we want to repeat the reasoning that led to equation \eqref{QME:modified} for the renormalized time-ordered product. To this end, we use the MWI \eqref{MWI} and replace $\Lap_{\textrm{nren}}$ with renormalized BV Laplacians. The anomaly term corresponding to the free action $S_0$ will be denoted by $\Lap({V})$ and the one of $S_0+\theta_0$ by $\tilde{\Lap}({\tilde{V}})$. The defining equation for $\tilde{\Lap}({\tilde{V}})$ is:
\be\label{MWI3}
\{e_{\sst{\TTR}}^{i\tilde{V}/\hbar},S_0+\theta_0\}_\star=\{ e_{\sst{\TTR}}^{i\tilde{V}/\hbar},S_0+\theta_0\}_{\sst{\TTR}}+e_{\sst{\TTR}}^{i\tilde{V}/\hbar}\TR(\tilde{\Lap}({\tilde{V}})+\frac{i}{2\hbar}\{\tilde{V},\tilde{V}\}_{\sst{\TTR}})\,.
\ee
The existence and properties of the anomaly term $\tilde{\Lap}({\tilde{V}})$ were proven in the case of Yang-Mills theory in \cite{H}. Analogous arguments can be also used in a more general setting. A short calculation yields
\begin{multline*}
e_{\sst{\TTR}}^{-iV/\hbar}\TR\left(\{e_{\sst{\TTR}}^{i V/\hbar},S_0\}_{\star}\right)=e_{\sst{\TTR}}^{-i\tilde{V}/\hbar}\TR\left(\{e_{\sst{\TTR}}^{i \tilde{V}/\hbar},S_0\}_{\star}\right)+\\+e_{\sst{\TTR}}^{-i\theta_0/\hbar}\TR\left(\{e_{\sst{\TTR}}^{i+\theta_0/\hbar},S_0\}_{\star}\right)+\Lap({V})-\tilde{\Lap}({\tilde{V}})\,.
 \end{multline*}
This shows that, if the free QME holds, and $\Lap({V+\theta_0})-\tilde{\Lap}(\tilde{V})$ can be removed with an appropriate redefinition of renormalized time-ordered products, then $e_{\sst{\TTR}}^{-iV/\hbar}\TR\left(\{e_{\sst{\TTR}}^{i V/\hbar},S_0\}_{\star}\right)=0$ is equivalent to $e_{\sst{\TTR}}^{-i\tilde{V}/\hbar}\TR\left(\{e_{\sst{\TTR}}^{i \tilde{V}/\hbar},S_0+\theta_0\}_{\star}\right)=0$. %\marginpar{when is it possible?}
\section{BRST charges}
\subsection{Different notions of a BRST charge}
We start this section with an overview of different approaches to quantization of gauge theories and different notions of a BRST charge in {\paqft}. We compare the approaches of \cite{DF99} and \cite{H}, using the general BV quantization framework proposed in \cite{FR3}. There is an important difference between the free BRST operator $\gamma_0$ used by Hollands in \cite{H} and the one discussed in \cite{FR}. The full BV operator 
$s$ is the same, but in \cite{FR,FR3} (following \cite{Barnich:1999cy}) $s_0$ is expanded with respect to the total antifield number $\#\ta$ (see equation \eqref{ta:expansion}), while in 
\cite{H} $s_0$ is expanded with respect to $\#\af$ \textit{also for the gauge fixed theory}. Explicitly, one has $s_0=\tilde{\delta}_0+\tilde{\gamma}_0$, where $\tilde{\gamma}_0$ is the term with $\#\af=0$ and $\tilde{\delta}_0$ has  $\#\af<0$.
We argue that the expansion with respect to $\#\ta$  is physically more justified, since the $\#\ta=-1$ term of this expansion, denoted here by $\delta$, is the Koszul operator corresponding to the gauge fixed system of equation of motion. This allows one to view the algebra of on-shell functionals as the 0-th cohomolgy of $(\BV,\delta)$; this interpretation is not possible if one considers the expansion used in \cite{H}. It is, therefore, not clear how going on-shell in \cite{H} is interpreted in cohomological terms, already at the level of linearized equations of motion. Let $C^\ddagger$ be the antifield of the ghost. Following \cite{H} we obtain $\tilde{\delta_0}(C^\ddagger)=id*d\bar{C}_I-dA^\ddagger_I$, so the equations of motion corresponding to the image of $\tilde{\delta_0}$ contain a source term $dA^\ddagger_I$ which is not present in the equations of motion employed in  \cite{H}  for the construction of the causal propagator and the star product. This problem is not present in \cite{FR,FR3}, where we have ${\delta_0}(C^\ddagger)=id*d\bar{C}_I$ and ${\gamma_0}(C^\ddagger)=-dA^\ddagger_I$, in contrast to ${\tilde{\gamma}_0}(C^\ddagger)=0$.

%Let us now discuss the properties of the free BRST charge. 
%It is important to realize that the version of the BV formalism proposed in \cite{FR,FR3} makes use of various products and Poisson brackets and the physical information is encoded in relations between these algebraic  structures. 
It was shown in \cite{FR} that, already at the classical level, we have two (graded) Poisson brackets: the Peierls bracket $\lfloor.,.\rfloor$ and the antibracket $\{.,.\}$. The Peierls bracket is induced by the equations of motion (dynamics) and the 
antibracket is of geometrical nature: it is the graded Schouten bracket. The BV operator $s$ can locally be written as the antibracket with the extended action $S_\ex$. In particular, the BRST operator is generated by $\theta$. Note that one uses $\{.,.\}$, rather than $\lfloor.,.\rfloor$, as the natural graded Poisson structure on $\BV(M)$. This can be compared with the Hamiltonian version of the Batalin-Vilkovisky formalism (BFV formalism, \cite{Batalin:1977pb,Fradkin:1977wv,Fradkin:1975cq,Fradkin:1977hw,Fradkin:1978xi}) where everything is done with relation to one canonical structure. It was shown in \cite{FH2,BGPR} that the Lagrangian and the Hamiltonian formalism are equivalent on the formal level. We believe that a more rigorous argument can be provided within the framework of infinite dimensional symplectic geometry. \cite{BGPR} compares also the Noether charge of the Lagrangian formalism with the BRST  charge of the Hamiltonian formalism. The Noether charge can be used as a generator of the BRST transformation with respect to a certain canonical structure which is defined on both fields and antifields (see formula 3.14 of \cite{BGPR}). In \cite{H} antifields are treated as external fields, so there is no dynamics associated with them and $\lfloor.,.\rfloor$ acts on them trivially. Also in \cite{FR} antifields are non-dynamical, since they are identified with geometrical objects: functional derivatives. Because $\lfloor.,.\rfloor$ is antifield-independent, the classical BRST charge $Q$ generates $\gamma$ with respect to $\lfloor.,.\rfloor$ modulo the equations of motion only on the space of functionals that don't contain antifields.  Let us briefly recall the construction of $Q$. The classical BRST current is defined as 
\begin{multline*}
J^\mu(x)\doteq\sum\limits_{\al}\Big(\gamma\ph^\al\frac{\partial{L_M(x)}}{\partial(\nabla_\mu\ph^\al)}+2\nabla_\nu\gamma\ph^\al\frac{\partial{L_M(x)}}{\partial(\nabla_\mu\nabla_\nu\ph^\al)}+\\
-\nabla_\nu\left(\gamma\ph^\al\frac{\partial{L_M(x)}}{\partial(\nabla_\mu\nabla_\nu\phi^\al)}\right)\Big)+-K^\mu_{M}(x)\,,
\end{multline*}
where $K_{M}^\mu$ is the divergence term appearing after applying $\gamma$ to $L_M(f)$. Following \cite{H} we recall here a useful formula relating $J$ with the BV operator:
\be\label{dJ}
dJ(x)=\sum_\al\{L_\ex(f),\ph^\al(x)\}\cdot\{\ph^\dgr_\al(x),L_\ex(f)\}=\sum_\al \theta^\al(x)\cdot\frac{\delta L_\ex(f)}{\delta\ph^\al(x)}\,,
\ee
where $f(x)=1$. 
\subsection{The free BRST charge}\label{free:charge}
%The free BRST charge becomes important in the construction of states. It was shown in \cite{DF99} that states for the interacting theory can be defined by deformation of states of the free theory, which are obtained as the cohomology of the free BRST charge.
%@ write something about perturbative construction of a state.
We have already shown in \ref{gauge:inv} that if $\theta_0$ is included into the free action, then additional consistency condition are needed. In particular, \eqref{const:cond2} has to hold and the background configuration $\ph_0$ has to be a solution of the equations of motion. Here, we show that the same conditions allow us to express $\{.,\theta_0\}_\star$ as the commutator with the free BRST charge. 
%Under conditions discussed in \ref{gauge:inv}, $S_0$ is invariant under $\gamma_0$ and can be quantized with the use of Kugo-Ojima formalism \cite{KuOji0,KuOji}. Such an approach has already been taken in \cite{DF99} in the study of QED.  In \cite{H} a similar formalism was applied to Yang-Mills theories, but different Ward identities were postulated. 
Let the free BRST current  be denoted by $J_0$.
%\[
%J_0^\mu=\sqrt{-g}g^{\mu\la}(B_\rho\nabla_\la C^\rho-(\nabla_\la B_\rho) C^\rho) \,.
%\]
In a spacetime $M$ with compact Cauchy surface $\Sigma$ there exists a closed compactly supported 1-form $\al$ on $M$ such that $\int_M\al\wedge\beta=\int_\Sigma\beta$, for any closed 3-form $\beta$. In this case, we can define the free BRST charge by
 \[
 Q_0\doteq \int_M\al\wedge J_0\,.
 \]
 %In \cite{DF99} the free quantum BRST operator is defined as the commutator $[.,Q_0]_\star$.   
 In \cite{H} the free quantum BRST operator is defined  directly by giving its action on basic fields and requiring that it is a $\star$-derivation. In the formalism of \cite{FR3} this corresponds to defining the free quantum BRST as $\{.,\tilde{\theta}_0\}_\star$, where  $\tilde{\theta}_0$ denotes the action that generates $\tilde{\gamma}_0$. The difference between 
 \cite{H} and \cite{FR3} lies again in the way in which   the free quantum BRST operator acts on antifields. In \cite{H} we have  $\{\ph^\ddagger_\al,\tilde{\theta}_0\}=0$ for all $\ph_\al^\ddagger$, whereas the formalism of \cite{FR3} applied to Yang-Mills theory yields $\{C^\ddagger,\tilde{\theta}_0\}=-dA^\ddagger$. In a proposition below we show that $\{.,\theta_0\}_\star$ is on-shell equal to $[.,Q_0]_\star$, if the argument has $\#\ta=0$ (i.e. it doesn't contain antifields). In general, however, $Q_0$ is not a generator for $\{.,\theta_0\}_\star$. This does not pose a problem, since Ward identities in \cite{FR3} are formulated in terms of $\{.,\theta_0\}_\star$, not $[.,Q_0]_\star$, so the results of \cite{H} can be applied. Note that $Q_0$ is not necessary for the construction of the abstract net of interacting algebras of observables. It is, however, a crucial concept in the Kugo-Ojima formalism \cite{KuOji0,KuOji}, which is a convenient method to construct states for the free theory. A deformation procedure given in \cite{DF99} allows then to construct states also on the net of local algebras of observables of the interacting theory.
% 
%We will now show that the commutator with the free BRST charge implements on-shell the derivation $\{.,\theta_0\}_\star$ of the quantum algebra  on the space functionals with $\#\ta=0$, if the consistency conditions \eqref{const:cond2} hold. We assume here that $\gamma_0^2=0$. This implies that the term $\{\theta_0,\theta_0\}$ can be neglected in the definition of the conserved current $J_0$. 

From now on we work in the algebraic adiabatic limit, which means that we are interested only in constructing local algebras $\fA(\Ocal)$ and don't discuss the existence of the inductive limit. Therefore, we can apply the idea of \cite{DF99} and embed $\Ocal$ into a spacetime with a compact Cauchy surface, for example into a causal completion of a spacial box. Keeping this in mind we restrict our attention to the situation where $M$ has a compact Cauchy surface.
%We need here an additional technical assumption. If the action of symmetries is well defined off-shell, then
%by definition we have $\gamma^2=0$, so $\{{\theta}_M,{\theta}_M\}\sim 0$ follows, where the relation $\sim$ is defined by \ref{equ}. One can, however, impose a stronger condition on the natural transformation $\theta$ generating $\gamma$.
%\begin{ass}\label{as:theta}
%Assume that we can choose a natural Lagrangian $\theta$ generating $\gamma$ in such a way that
%\[
%\{\theta_0(f),\theta_0(f)\}=0\quad\forall f\in\D(M)\,.
%\]
%\end{ass}
%This assumption can be fulfilled in QED, Yang-Mills theories and gravity, so it doesn't seem too restrictive from the physical point of view. An example Lagrangian for the Yang-Mills case is
%\begin{multline*}
%\theta_M(f)=\int_M \big(d(fC)+\frac{1}{2}[A,fC]\big)^I_\mu(x)\frac{\delta}{\delta A^I_\mu(x)}+\\
%+\frac{1}{2}\int_M f [C,C]^I(x)\frac{\delta}{\delta C^I(x)}-i\int_M f B_I(x)\frac{\delta}{\delta \bar{C}_I(x)}\,.
%\end{multline*}
%In order to avoid this assumption one should,  instead of working on the level of functionals, work on the level of fields understood as natural transformations. %We discuss this possibility in detail in section \ref{}. %@ include an additional section.
\begin{prop}\label{Q0}
Let $F\in\BV(M)$ with $\#\ta=0$. Assume that \eqref{PK} and \eqref{const:cond2} hold and that $\gamma_0^2=0$, then the following relation is fulfilled on-shell ($\ \os\,$):
\[
\{F,\theta_0\}_{\star_H}\os\frac{i}{\hbar}[F,Q_0]_{\star_H}\,.
\]
Equivalently, we can write this formula in terms of Wick-ordered expressions $\al_H^{-1}(F)$, $\al_H^{-1}(Q_0)\in \fA(M)$,
\[
\{\al_H^{-1} F,\theta_0\}_{\star}\os\frac{i}{\hbar}[\al_H^{-1}F,\al_H^{-1}Q_0]_{\star}\,
\]
\end{prop}
\begin{proof}
Since we assume that $\gamma_0^2=0$, the term $\{\theta_0,\theta_0\}$ can be neglected in the definition of the conserved current $J_0$.
\[
dJ_0(x)=\sum_\al \theta_0^\al(f)(x)\cdot\frac{\delta L_0(f)}{\delta\ph^\al(x)}\,,
\]
where $f(x)=1$. Since the definition of the BRST charge doesn't depend on the choice of a 1-form $\al$ dual to the Cauchy surface, we can choose it in a way that will facilitate the calculation. Let us take $\al=d\eta$, where $\al$ is compactly supported, its support lies in the past of the support of $F$
and $\eta=1$ on $\supp\, F$. Now we use the fact that, for the Hadamard function $\omega=$ used to define the $\star$-product, $  \omega^{\alpha \beta}(x,y) - (-1)^{|\phi^\alpha| |\phi^\beta|} \omega^{\beta \alpha}(y,x)=i\Delta^{\al\bet}(x,y)$ holds. Moreover, from the support properties of $\Delta_R$ and $\Delta_A$ follows that we can write the commutator with $Q_0$ as
\begin{align*}
[F,Q_0]_{\star_H}&= i\hbar\sum_{\al,\bet,\sigma}\int dxdz\frac{\delta F}{\delta\ph^\bet(x)}\Delta_R^{\bet\al}(x,z) \frac{\delta Q_0(d\eta)}{\delta\ph^\al(z)}+\\
&+ i\hbar^2\sum_{\al,\bet,\sigma,\atop \mu,\nu}\int dxdz\frac{\delta^2 F}{\delta\ph^\bet(x)\ph^\mu(x')}\Delta_R^{\mu\nu}(x',z')H^{\bet\al}(x,z) \frac{\delta^2 Q_0(d\eta)}{\delta\ph^\al(z)\delta\ph^\nu(z')}\,.
\end{align*}
Now let us choose $f\in\D(M)$ such that $f\equiv 1$ on the support of $F$. Inserting the definition of $Q_0$ and integrating by parts we obtain:
\begin{align*}
\frac{i}{\hbar}[F,Q_0]_{\star_H}=&\sum_{\al,\bet,\sigma}\int\frac{\delta F}{\delta\ph^\bet(x)}\Delta_R^{\al\bet}(x,z)\eta(y) \frac{\delta \theta^{\sigma}_0(f)(y)}{\delta\ph^\al(z)}\frac{\delta S_0}{\delta\ph^\sigma(y)} dxdydz\\
&\sum_{\al,\bet\sigma}\int\frac{\delta F}{\delta\ph^\bet(x)}\Delta_R^{\al\bet}(x,z)\eta(y)\theta^{\sigma}_0(f)(y)\frac{\delta^2 S_0}{\delta\ph^\sigma(y)\delta\ph^\al(z)} dxdydz+\\
&\hbar\sum_{\al,\bet\sigma}\int\frac{\delta^2 F}{\delta\ph^\sigma(y)\delta\ph^\mu(x)}H^{\mu\nu}(x,z)\frac{\delta \theta^{\sigma}_0(f)(y)}{\delta\ph^\nu(z)} dxdydz\,.
\end{align*}
Note that, to obtain the above result, we had to exchange the order of integration and differentiation. This is possible only if we assume
 the consistency condition \eqref{const:cond3}. Without this assumption, the third term of the formula above would not be well defined for local $F$. After performing the integration over $y$ and $z$ in the second term, we arrive at:
\begin{multline*}
\frac{i}{\hbar}[F,Q_0]_{\star_H}=\sum_{\bet}\int\frac{\delta F}{\delta\ph^\bet(x)}\theta^{\bet}_0(x) dx+\\
+\hbar\sum_{\al,\bet\sigma}\int\frac{\delta^2 F}{\delta\ph^\sigma(y)\delta\ph^\mu(x)}H^{\mu\nu}(x,z)\frac{\delta \theta^{\sigma}_0(f)(y)}{\delta\ph^\nu(z)} dxdydz+I_0=\\
=\{F,\theta_0\}_{\star_H}+I_0\,,
\end{multline*}
where $I_0$ is an element of the ideal generated by equations of motion. 
%Assumption \ref{as:theta} guarantees that the term involving $\frac{\delta \theta_0}{\delta\phi^\al(x)}$ does not contribute. 
%A weaker condition, 
%\[
%\sum_\al K^\gamma_{\ \al}  K^\al_{\ \sigma} \Phi^\sigma\os0\,,
%\]
% holds automatically, as $\gamma_0^2= 0$ modulo the image of $\delta_0$ (i.e. modulo the equations of motion). 
% Assumption \eqref{KmK}
% is stronger and is fulfilled in particular for theories with independent (not reducible) gauge symmetries. For such theories one does not require to introduce ghosts of ghosts, so $C(x)$ is the basic field with highest $\#\gh$.  $\theta_0$ has to be of degree $+1$, so the only nontrivial contribution to $\frac{\delta\theta_0}{\delta\phi^\al(x)}$ comes from $\phi^\al(x)=C(x)$ or $\phi^\al(x)=B(x)$, if the non-minimal sector is included. We know however that the free BRST transformation acts trivially on ghosts, so \eqref{KmK} holds. 
\end{proof}
%Using the above proposition and the MWI \eqref{MWI2} we obtain a relation
%\begin{align}\label{T12a}
%[e^{iV/\hbar}_{\sst{\TT}},Q_0]_\star&\os-i\hbar\{e_{\sst{\TT}}^{iV/\hbar},\theta_0+S_0\}_\star=\nonumber\\
%&=e_{\sst{\TT}}^{i/\hbar}\T\left(\{\theta_0+S_0,V\}_{\sst{\TT}}+\frac{1}{2}\{V,V\}_{\sst{\TT}}-i\hbar\tilde{\Lap}_V\right)\,,
%\end{align}
%In the second step we added an on-shell term $\{e_{\sst{\TT}}^{iV/\hbar},S_0\}_\star$. Proposition \ref{Q0} together with the QME allow us also to write
%\be\label{RV:Q0}
%[R_V(F),Q_0]_\star\os R_V\left(\{\theta_0+S_0,V\}_{\sst{\TT}}+\frac{1}{2}\{V,V\}_{\sst{\TT}}-i\hbar\tilde{\Lap}_V\right)\,.
%\ee
%Note that $R_V^{-1}$ intertwines the free quantum theory with the interacting one and, in particular, $R_V^{-1}\left(\tfrac{\delta S_0}{\ph(x)}\right)=\tfrac{\delta S_0}{\ph(x)}+\tfrac{\delta V}{\ph(x)}$\marginpar{check!}. The equations of motion ideal of the interacting theory is, therefore, generated by $R_V^{-1}\left(\tfrac{\delta S_0}{\ph(x)}\right)$ and equation \eqref{RV:Q0} can be stated in the form:
%\[
%\{\theta_0+S_0,V\}_{\sst{\TT}}+\frac{1}{2}\{V,V\}_{\sst{\TT}}-i\hbar\tilde{\Lap}_V=R_V^{-1}[R_V(F),Q_0]_\star=[F,R_V^{-1}(Q_0)]_{\star_V}\quad \textrm{mod interacting EOM's}\,. 
%\]
The result above allows to make contact with the formalism used in \cite{DF99} and \cite{Boas}, where gauge theories are quantized in the BRST formalism, but without introducing antifields. The quantum BRST differential on free fields is defined as the commutator with $Q_0$ and, as we have just seen, this is the same as $\{.,\theta_0\}_{\star_H}$ on the space of functionals with $\#\ta=0$. 
 \subsection{The interacting BRST charge}
Up to now we have treated only the free theory, now we want to construct a charge that generates the quantum BV differential on interacting fields. We also want to drop the assumption on the total antifield number $\#\ta$. We work in the algebraic adiabatic limit, so we pick a bounded region $\Ocal\subset M$ and choose $f\in\D(M)$ such that $f\equiv 1$ on $\Ocal$. We choose $\al=d\eta$ as in the previous section, and require that $\supp(\al)\subset\Ocal$; let $V={S_1}_M(f)$. Interacting fields are defined by formula \eqref{RV}. For such fields we cannot make use of proposition \ref{Q0}, because $V$ in general contains antifields and $\eta$ cannot be chosen to be one on the support of $R_V(F)$. Instead, we can follow \cite{DF99,Boas,H} and use the interacting charge $R_V(Q)$.
Note that our situation is much more general than the cases studied in the literature so far. In the present setting we admit arbitrary theories with local symmetries, which satisfy the {\qme}. These include in particular gravity and the free bosonic string. 

Firstly, we need to prove some identities for time-ordered retarded products. 
For simplicity of notation we omit the subscript ``$M$'' in ${\theta}_M$, ${S_1}_M$, etc. The conservation of the current $R_V(J(x))$ is a condition that $R_V(dJ(x))\os 0$ for $x\in\Ocal$. We have to prove that
\be\label{current:cons}
 e_{\sst \TT}^{iV/\hbar}\T dJ(x)=e_{\sst \TT}^{iV/\hbar}\T \left(\theta^\al(x)\cdot\frac{\delta}{\delta \ph^\al(x)}(S_0+S_1)(f)\right)\os 0\,.
\ee
Let $h\in\D(\Ocal)$. Using the {\mwi} \eqref{Lap:coeff} we obtain an identity fulfilled by the natural transformation $dJ$: 
\begin{multline*}
 e_{\sst \TT}^{iV/\hbar}\T dJ(h)=-i\hbar\int h(x)\left( e_{\sst \TT}^{iV/\hbar}\T \theta^\al(x)\right)\star\frac{\delta S_0}{\delta\ph^\al(x)}d\mu(x)+\\
 +i\hbar\sum_{n=0}^\infty\Lap^{n}(V(f)^{\otimes n}; V(h))\,,
\end{multline*}
If the ``anomaly'' $\sum_{n=0}^\infty\Lap^{n}(V(f)^{\otimes n}; V(h))$ can be removed by a redefinition of time-ordered products, then the above identity implies \eqref{current:cons}. In \cite{H} the anomaly was removed in Yang-Mills theory. In general the interacting current is not conserved in theories which satisfy the QME with a non-zero $\Lap({V(f)})$ that cannot be removed. %\marginpar{are there examples of such theories?}. 
It is an intuitive result, since current conservation is a classical phenomenon and to reproduce it on the quantum level one  has to assume that the quantized theory ``doesn't differ too much'' from the classical one. 
%Assuming \ref{as:theta} and the QME\footnote{In \cite{H} it was shown that, for the gauge fixed Lagrangian $L$ of the Yang-Mills theory,  one can redefine the time-ordered products in such a way that $L_M(f)$, $f\in\D(M)$ satisfies the QME with $\Lap_V=0$. Similar techniques can be used also in perturbative quantization of gavity.} for $V$,
%\begin{align*}
%e_{\sst \TT}^{V}\T Q&= e_{\sst \TT}^{V}\T\!\!\int   \eta \wedge dJ=e_{\sst \TT}^{V}\T \int\!\! \eta(x) \theta^\al(x)_\cdot\frac{\delta}{\delta \ph^\al(x)}(S_0+V)d\mu(x)=\\
%&=e_{\sst \TT}^{V}\T\left(\Big\{\theta(f\eta),S_0\Big\}_{\sst \TT}+\int\! \eta\, \theta^\al(f)\frac{\delta V}{\delta\ph^\al}d\mu\right)=\\
%&=\Big\{e_{\sst \TT}^{V}\T \theta(f\eta),S_0\Big\}_{\star}+e_{\sst \TT}^{V}\T\left(\int\! \eta\, \theta^\al(f)\frac{\delta V}{\delta\ph^\al}d\mu-\{V,\theta(f\eta)\}_{\sst\TT}+i\hbar\Lap_{V}(\theta(f\eta))\right)\,.
%\end{align*}
%In the last step we made use of \eqref{MWI2}. 

Recall that the map $R_V$ intertwines between the interacting and the free theory and in particular we have:
\[
R_V^{-1}\left(\frac{\delta S_0}{\delta\ph^\al(x)}\right)=\frac{\delta (S_0+V)}{\delta\ph^\al(x)}\,.
\]
Let us denote by $\osV$ an equality that holds modulo the ideal generated, with respect to $\star_V$, by $R_V^{-1}\left(\frac{\delta S_0}{\delta\ph^\al(x)}\right)$. We are now ready to prove the main result of this paper.
\begin{thm}
Assume that the {\qme} holds for $V\in \mathscr{V}_{S_1}(\Ocal) $ and that $\Lap(V)=0$. Let $F\in\fA_\loc(\Ocal)$, then
\be\label{Q:int}
\frac{i}{\hbar}[R_V(F),R_V(Q)]_\star\os R_V( s F-i\hbar\Lap_V(F))=R_V(\hat{s} F)\,.
\ee
In other words, $Q$ generates, with respect to $[.,.]_{\star_V}$, the quantum BV operator,
\be\label{Q:int2}
\frac{i}{\hbar}[F,Q]_{\star_V}\osV  \hat{s} F\,.
\ee
\begin{proof}
Using the GLZ relation \eqref{glz} we obtain
\[
\frac{1}{i\hbar}[R_V(F),R_V(Q)]_\star=R_V^{(1)}[F](Q)-R_V^{(1)}[Q](F)\,.
\]
We can always write a local functional $F\in\fA_\loc(\Ocal)$ in the form $F=\int F(x) d\mu(x)$. Using \eqref{current:cons}, \eqref{MWI2} and assuming that $\Lap({V(f)})$  was already removed, we find that
\[
R_V^{(1)}[dJ_0(\eta)](F)=\frac{d}{d\la}(e^{iV/\hbar}_{\sst \TT})^{-1\star}\star\left(e^{i(V+\la F)/\hbar}_{\sst \TT}\T dJ(\eta)\right)\Big|_{\la=0}\,,
\]
and we can use relation \eqref{MWI} to obtain
\begin{multline*}
e^{i(V+\la F)/\hbar}_{\sst \TT}\T dJ(\eta)=-i\hbar\int \eta(x)\left(e^{i(V+\la F)/\hbar}_{\sst \TT}\T \frac{\delta (\theta^\al+\la F)}{\delta\ph_\al^\ddagger(x)}\right)\star\frac{\delta S_0}{\delta \ph^\al(x)}d\mu(x)+\\
+i\hbar\sum_{n=0}^\infty\Lap^{n}((V(f)+\la F)^{\otimes n}; V(\eta)+\la F(\eta))+\\
-e^{i(V+\la F)/\hbar}_{\sst \TT}\T\left(\{\la F(\eta),S_0+V\}_{\sst \TT}+\la^2\int \frac{\delta F(\eta)}{\delta \ph_\al^\ddagger(x)}\T\frac{\delta F}{\delta \ph^\al(x)}\right)\,.
\end{multline*}
where $F(\eta)\doteq \int F(x)\eta(x) d\mu(x)$. Thus we obtain
\begin{multline}
R_V^{(1)}[dJ(\eta)](F)\os - \, R_V(\{F(\eta),S_\ex\}_{\sst \TT})+\\
+i\hbar R_V\left(\sum_{n=1}^\infty n\Lap^{n}\left(V(f)^{\otimes(n-1)}\otimes F;\theta(\eta)\right)\right)+\\
+i\hbar R_V\left(\sum_{n=1}^\infty \Lap^{n}\left(V(f)^{\otimes n};F(\eta)\right)\right)\,.\label{Wid}
\end{multline}
Similarly for $R_V^{(1)}[F](dJ_0(\eta))$. 

Now, we can make a particular choice for the function $\eta$ in the definition of $Q_0$. The current conservation implies that $Q_0$ is independent of this choice. We are interested in local algebras, so we can assume that $\Ocal$ is embedded in a spacetime $M$ with a compact Cauchy surface $\Sigma$. We pick two other Cauchy surfaces $\Sigma_{\pm}$, such that $\Sigma_-$ is in the past of   $\supp(F)$ and $\Sigma_+$ in its future. We choose a function $\eta$ such that, for any closed 3-form $\beta$, $\int_Md\eta\wedge\beta=\int_\Sigma\beta$ holds. Next, we take compactly supported functions $\eta_\pm$ such that $d\eta_\pm=d\eta+\chi_\pm$, where $\chi_\pm$ are supported in the future (past) of $\Sigma_+$ ($\Sigma_-$). Moreover we require that, on $\supp(F)$, $\eta_+=\eta$ and $\eta_-=\eta-1$ hold. An explicit construction of such functions is provided in \cite{DF99} for a flat $M$. Using the support property \eqref{supp:prop} of retarded products we find that
\begin{align*}
R_V^{(1)}[Q](F)-R_V^{(1)}[F](Q)&=-R_V^{(1)}\Big[\int_Md\eta_+\wedge J\Big](F)+R_V^{(1)}[F]\Big(\int_Md\eta_-\wedge J\Big)=\\
&=R_V^{(1)}\left[dJ(\eta_+)\right](F)-R_V^{(1)}[F]\left(dJ(\eta_-)\right)\,.
\end{align*}
Inserting \eqref{Wid} and using the properties of $\eta_\pm$ on the support of $F$ we obtain
\begin{align*}
[R_V(F),R_V(Q)]_\star&= -i\hbar\, (R_V(\{F,S_\ex\}_{\sst \TT})-i\hbar R_V\left(\Lap_V(F)\right))=\\
&=-i\hbar R_V(\hat{s}F)\,.
\end{align*}
\end{proof}
\end{thm}
Note that our result doesn't require $\Lap_V(F)$ to vanish. This means that, on-shell, $Q$ is a generator for the full \textit{quantum} BV operator. It is interesting to ask what modifications have to be made to allow for the situation in which $\Lap(V)$ is also non-zero. We consider it as a problem for future study.
\section*{Acknowledgements}
Most of the results presented in this paper were obtained during a Junior Hausdorff Trimester Program that took place during the fall 2012 at the Hausdorff Research Institute for Mathematics (HIM) in Bonn. Therefore, I would like to thank the participants of the program and the guests of HIM with whom I discussed during my stay. In particular, I gained from discussions with C. Dappiaggi, W. Dybalski, Ch. Fewster, K. Fredenhagen, S. Meinhardt, J. Schlemmer, Y. Tanimoto, M. Wrochna and J. Zahn (whom I also thank for important remarks on the first version of the manuscript). 

\end{document}